%% file: smidpoint.tex
\documentclass{scrartcl}

\usepackage[
final,
theoremdefs
]{allergy}
\usetikzlibrary{calc}
\usepackage{unixode}
\usepackage{helvetic}
\addtokomafont{sectioning}{\color{DarkRed}}
\addtokomafont{title}{\color{DarkSlateGrey}}

\usepackage{enumitem}
\usepackage{mdframed}

\usepackage[numbers]{natbib} 

\usetikzlibrary{calc}
\usetikzlibrary{matrix}
\tikzstyle{commdiag}=[matrix of math nodes, row sep=3em, column sep=5.5em, text height=1.5ex, text depth=0.25ex,ampersand replacement=\&]
\tikzset{>=stealth}

\usepackage{standalone}

\usepackage{subfig}

\usepackage[affil-sl]{authblk} 

\title{A minimal-variable symplectic integrator on spheres}
\author[1]{Robert McLachlan\thanks{\emaillink{r.mclachlan@massey.ac.nz}}}
\author[2]{Klas Modin\thanks{\emaillink{klas.modin@chalmers.se}}}
\author[3]{Olivier Verdier\thanks{\emaillink{olivier.verdier@hvl.no}}}

\affil[1]{
	Institute of Fundamental Sciences, Massey University, New Zealand
}
\affil[2]{
	Mathematical Sciences, Chalmers University of Technology, Sweden
}
\affil[3]{
  Department of Computing, Mathematics and Physics, Western Norway University of Applied Sciences, Bergen, Norway
}

\date{\today}

\input{definitions}

\newcommand\RR{\mathbf{R}}

\newcommand\bw{\boldsymbol{w}}

\newcommand{\w}{\vect{w}}

\newcommand{\midp}[1]{\overline{#1}}
\newcommand{\difp}[1]{\overrightarrow{#1}}

\newcommand{\wminus}{\vect{w}}
\newcommand{\wplus}{\vect{W}}
\newcommand{\wmid}{\midp{\wminus}}
\newcommand{\wdif}{\difp{\wminus}}

\newcommand{\aneutral}{\sigma}
\newcommand{\aminus}{\sigma}
\newcommand{\aplus}{\Sigma}
\newcommand{\amid}{\midp{\aminus}}
\newcommand{\adif}{\difp{\aminus}}

\newcommand{\bneutral}{\lambda}
\newcommand{\bminus}{\lambda}
\newcommand{\bplus}{\Lambda}
\newcommand{\bmid}{\midp{\bminus}}
\newcommand{\bdif}{\difp{\bminus}}

\NewDocumentCommand\HH{s}{\mathbf{H}\IfBooleanTF{#1}{_*}{}}
\NewDocumentCommand\Hb{s}{{V}\IfBooleanTF{#1}{_*}{}}
\RenewDocumentCommand\CC{s}{\mathbf{C}\IfBooleanTF{#1}{_*}{}}
\NewDocumentCommand\conj{O{z}}{\overline{\vect #1}}

\NewDocumentCommand\projS{}{\rho}

\def\u{{\boldsymbol{u}}}

\def\f{X}

\begin{document}

\maketitle

\begin{abstract}
We construct a symplectic, globally defined, minimal-variable, equivariant integrator on products of 2-spheres.
Examples of corresponding Hamiltonian systems, called spin systems, include the reduced free rigid body, the motion of point vortices on a sphere, and the classical Heisenberg spin chain, a spatial discretisation of the Landau--Lifshitz equation.
The existence of such an integrator is remarkable, as the sphere is neither a vector space, nor a cotangent bundle, has no global coordinate chart, and its symplectic form is not even exact.
Moreover, the formulation of the integrator is very simple, and resembles the geodesic midpoint method, although the latter is \emph{not} symplectic.
\end{abstract}

\section{Introduction} 
\label{sec:introduction}

The 2--sphere, denoted $S^2$, is a fundamental symplectic manifold that occurs as the phase space, or part of the phase space, of many Hamiltonian systems in mathematical physics.
A globally defined symplectic integrator on $S^2$ needs a minimum of three variables, since the lowest-dimensional vector space in which $S^2$ can be embedded is $\RR^3$.
To construct such a minimal-variable, symplectic integrator is, however, surprisingly difficult, and has long been an open problem. Here we solve that problem.
We equip the direct product of $n$ 2-spheres, $(S^2)^n$, with the symplectic form~$\omega$ given by the weighted sum of the area forms 
\begin{equation}\label{eq:symp_form}
	\omega = \sum_{i=1}^{n}\kappa_i \ud A_i, \quad \kappa_i > 0,
\end{equation}
where $\ud A_i$ is the standard area form on the $i$:th sphere.

Throughout the paper, we represent $S^2$ by the space of unitary vectors in $\RR^3$.
General Hamiltonian systems on~$(S^2)^n$ with respect to the symplectic form~\eqref{eq:symp_form} take the form
\begin{equation} 
	\label{eq:classspinsys}
	\dot{\w}_i = \w_i \times \frac{1}{\kappa_i}\frac{\pd H}{\pd \w_i},\quad \w_i\in S^2,\quad i=1,\ldots,n, \quad H\in C^{\infty}((S^2)^n).
\end{equation}
We provide a global, second order symplectic integrator for such systems, which we call the \emph{spherical midpoint method}.
The method is remarkably simple: for a Hamiltonian function $H\in C^{\infty}((S^2)^n)$, it is the mapping 
\begin{equation}
	\paren[\big]{S^2}^n\ni (\wminus_1,\ldots,\wminus_n)\mapsto\paren[\big]{\wplus_1,\ldots,\wplus_n}\in \paren[\big]{S^2}^n,
\end{equation}
defined by
\begin{mdframed}
\begin{equation}\label{eq:short_method}
	\frac{\wplus_i - \wminus_i}{h} = \frac{\wminus_i+\wplus_i}{\abs{\wminus_i+\wplus_i}}\times \frac{1}{\kappa_i}\frac{\pd H}{\pd \w_i}\left( \frac{\wminus_1+\wplus_1}{\abs{\wminus_1+\wplus_1}},\ldots,\frac{\wminus_n+\wplus_n}{\abs{\wminus_n+\wplus_n}} \right),
\end{equation}
\end{mdframed}
where $h>0$ is the step size.
In addition to be symplectic, this method is \emph{equivariant}, meaning it respects the intrinsic symmetries of the 2--sphere.
Put differently, it respects the homogeneous space structure $S^2\simeq \SO(3)/\SO(2)$, a property analogous to the affine equivariance of B-series methods~\cite{McMoMuVe2014}.
Note also, as we observe in \autoref{rk:geodeticmp}, that our method is \emph{not} the geodesic midpoint method applied to \eqref{eq:classspinsys}.

Equations of the form \eqref{eq:classspinsys} are called 
\emph{classical spin systems}~\cite{La2011}.
The simplest example is the reduced free rigid body
\begin{equation}
	\dot{\w} = \w \times \vect{I}^{-1} \w,\quad \w\in S^2.
\end{equation}
Other examples include the motion of massless particles in a divergence-free vector field on the sphere (for example, test particles in a global weather simulation), the motion of~$n$ point vortices in a ideal incompressible fluid on the sphere, and the set of Lie--Poisson systems on $\mathfrak{so}(3)^*$.
Spin systems with large $n$ are obtained by spatial discretisations of Hamiltonian PDEs on $S^2$.
An example is the classical Heisenberg spin chain of micromagnetics,
\begin{equation}
	\dot{\w}_i = \w_i \times (\w_{i+1} - 2\w_{i} + \w_{i-1}),\quad \w_0 = \w_n, \quad \w_i\in S^2,
\end{equation}
which is a spatial discretisation of the Landau--Lifshitz PDE
\begin{equation}\label{eq:landau_lifshitz}
	\dot{\w} = \w\times \w'', \quad \w\in C^{\infty}(S^1,S^2).
\end{equation}

Apart from its abundance in physics, there are a number of reasons for focusing on the phase space $(S^2)^n$. 
It is the first example of a symplectic manifold that
\begin{itemize} \item is \emph{not} a vector space, \item is \emph{not} a cotangent bundle, \item does \emph{not} have a global coordinate chart or a cover with one, and \item is \emph{not} exact (that is, the symplectic form is not exact). \end{itemize}
Furthermore, next to cotangent bundles, the two main types of symplectic manifolds are coadjoint orbits of Lie--Poisson manifolds and K\"ahler manifolds; $(S^2)^n$ is the simplest example of both of these. 

Lie group integrators for general systems on $(S^2)^n$ are developed in~\cite{LeNi03}.
These are, however, not symplectic.
Symplectic integrators for \emph{some} classical spin systems are given in \cite{StSc06,LuWaBr10}.
These are, however, based on splitting, and therefore not applicable for general Hamiltonians.

To find symplectic integrators on $(S^2)^n$ for general Hamiltonians is particularly challenging because
symplectic integrators for general Hamiltonians are closely related to the classical canonical \emph{generating functions} defined on \emph{symplectic vector spaces} (or in local canonical coordinates). 
Generating functions are a tool of vital importance in mechanics, used for perturbation theory, construction of orbits and of normal forms, in bifurcation theory, and elsewhere. 
They have retained their importance in the era of symplectic geometry and topology, being used to construct Lagrangian submanifolds and to count periodic orbits~\cite{We1981,Vi1994}. 
Although there are different types of generating function, all of them are restricted to cotangent bundle phase spaces.

In our case, the four `classical' generating functions, that treat the position and momentum differently, do not seem to be relevant given the symmetry of $S^2$.
Instead, our novel method (or generating function) is more related to the Poincaré generating function \cite[vol.~III, \S319]{Po1892} 
\begin{equation}
	J(\wplus-\wminus)=\nabla G\Big(\frac{\wplus+\wminus}{2}\Big), \qquad J = \begin{pmatrix}0 & I \cr -I & 0 \end{pmatrix},	
\end{equation}
which is equivariant with respect to the full affine group and which corresponds to the classical midpoint method when interpreted as a symplectic integrator.
The classical midpoint method on vector spaces is known to conserve quadratic invariants \cite{Co1987}, and hence automatically induces a map on $S^2$ when applied to spin systems.
However, it has long been known \emph{not} to be symplectic \cite{AuKrWa1993}. 

We now list the already known techniques to construct symplectic integrators for general Hamiltonian systems on a symplectic manifold $M$ that is not a vector space:
\begin{enumerate}
	\item If $M = T^*Q$ is the cotangent bundle of a submanifold $Q\subset \RR^{n}$ determined by level sets of $m$ functions $c_{1},\ldots,c_{m}$, then the family of \Rattle\ methods can be used~\cite{Ja1996}.
	More generally, if $M$ is a transverse submanifold of $\RR^{2n}$ defined by coisotropic constraints, then geometric \Rattle\ methods can be used~\cite{McMoVeWi2013}.

	\item If $M\subset \mathfrak{g}^{*}$ is a coadjoint orbit (symplectic leaf) of the dual of a Lie algebra $\mathfrak{g}$ corresponding to a Lie group $G$, \Rattle\ methods can again be used:
	first extend the symplectic system on $M$ to a Poisson system on $\mathfrak{g}^{*}$,
	then ``unreduce'' to a symplectic system on $T^*G$, then embed $G$ in a vector space and use strategy 1 above~\cite[\S\!~VII.5]{HaLuWa2006}.
	One can also use \emph{Lie group integrators} for the unreduced system on $T^*G$~\cite{ChSc1991,MaPeSh1999}.
	The discrete Lagrangian method, pioneered in this context by \citet{MoVe1991}, yields equivalent classes of methods. 
	The approach is very general, containing a number of choices, especially those of the embedding and the discrete Lagrangian. 
	For certain choices, in some cases, such as the free rigid body, the resulting discrete equations are completely integrable; this observation has been extensively developed~\cite{DeLiTo1992}.

	\item If $M\subset\mathfrak{g}^{*}$ is a coadjoint orbit and $\mathfrak{g}^{*}$ has a symplectic realisation on $\RR^{2n}$ obtained through a momentum map associated with a Hamiltonian action of $G$ on $\RR^{2n}$, then symplectic Runge--Kutta methods for \emph{collective Hamiltonian systems} (cf.~\cite{MaWe1983}) sometimes descend to symplectic methods on~$M$ (so far, the cases $\mathfrak{sl}(2)^{*},\su(n)^{*}$, $\so(n)^{*}$, and $\mathfrak{sp}(n)^{*}$ have been worked out).
	This approach leads to \emph{collective symplectic integrators}~\cite{McMoVe2013b}.
\end{enumerate}

Let us review these approaches for the case $M=S^{2}$.

The first approach is not applicable, since $S^{2}$ is not a cotangent bundle.


The second approach is possible, since $S^{2}$ is a coadjoint orbit of $\su(2)^{*}\simeq \RR^{3}$.
$\SU(2)$ can be embedded as a 3--sphere in $\RR^{4}$  using unit quaternions, which leads to methods that use $10$ variables, in the case of \Rattle\ ($8$ dynamical variables plus $2$ Lagrange multipliers), and $8$ variables, in the case of Lie group integrators.
Both of these methods are complicated; the first due to constraints and the second due to the exponential map and the need to solve nonlinear equations in auxiliary variables.

The third approach is investigated in~\cite{McMoVe2014c}.
It relies on a quadratic momentum map $\pi\colon T^*\RR^{2}\to \su(2)^{*}$ and integration of the system corresponding to the \emph{collective Hamiltonian} $H\circ \pi$ using a symplectic Runge--Kutta method.
This yields relatively simple integrators using $4$ variables.
They rely on an auxiliary structure (the suspension to $T^*\RR^{2}$ and the Poisson property of $\pi$) and requires solving nonlinear equations in auxiliary variables; although simple, they do not fully respect the simplicity of $S^2$.

Our spherical midpoint method, fully described in \autoref{sec:main_result}, is simpler than all of the known approaches above; it is as simple as the classical midpoint method on vector spaces.
We would like to emphasise, however, that symplecticity of our method is by no means related to the symplecticity of the classical midpoint method.
The existence of the spherical midpoint method is thus unexpected, and its symplecticity is surprisingly difficult to prove.

In \autoref{sec:example} we provide a series of detailed numerical examples for various spin systems.
Interestingly, the error constants for the spherical midpoint method appears to be significantly smaller than for the \Rattle\ method.


Finally, while the present study is phrased in the language of numerical integration, we wish to remind the reader of the strong relation to discrete time mechanics, a field studied for many reasons:
\begin{enumerate}[label=(\roman*)]
\item It has an immediate impact  in computational physics, where symplectic integrators are in widespread use and in many situations are overwhelmingly superior to standard numerical integration~\cite{McQu2006}. 
	\item As a generalisation of continuous mechanics, discrete geometric mechanics is in principle more involved: the nature of symmetries, integrals, and other geometric concepts is important to understand both in its own right and for its impact on numerical simulations~\cite{HaLuWa2006}. 
	\item Discretisation leads to interesting physics models, for example the extensively-studied Chirikov standard map~\cite{Ch2008}.
	\item Discrete models can also be directly relevant to intrinsically discrete situations, such as waves in crystal lattices.
Here, the appearance of new phenomena, not persisting at small or vanishing lattice spacing, is well known~\cite{Ka1993}.
\item The field of discrete integrability is undergoing rapid evolution, with many new examples, approaches, and connections to other branches of mathematics, e.g., special functions and representation theory~\cite{KoBaTa2004}.
\item A strand of research in physics, pioneered notably by \citet{Le1983}, develops the idea that time is fundamentally discrete, and it is the continuum models that are the approximation. 
\item Discrete models can contain ``more information and more symmetry than the corresponding differential equations''~\cite{Le1986}; this also occurs in discrete integrability~\cite{KoBaTa2004}. 
\end{enumerate}


\section{Main results} 
\label{sec:main_result}

We present our two methods, the spherical midpoint method, and the extended spherical midpoint method, and state their properties.

We use the following notation.
$\Xcal(M)$ denotes the space of smooth vector fields on a manifold $M$.
If $M$ is a Poisson manifold, and $H\in C^{\infty}(M)$ is a smooth function on $M$, then the corresponding Hamiltonian vector field is denoted $X_H$.
The Euclidean length of a vector $\w\in\RR^{d}$ is denoted $\abs{\w}$.
If $\w\in\RR^{3n}\simeq (\RR^{3})^n$, then $\w_i$ denotes the $i$:th component in $\RR^{3}$.

\subsection{Spherical Midpoint Method: Symplectic integrator on spheres}


Our paper is devoted to the following novel method.

\newcommand\mw{\paren{ \wminus + \wplus}}

\begin{definition}\label{def:spherical_midpoint}
	The \textbf{spherical midpoint method} for $\xi\in\Xcal\big((S^{2})^{n}\big)$ is the numerical integrator
	\begin{equation}\label{eq:area_midpoint}
		\Phi(h\xi)\colon (S^{2})^{n}\to (S^{2})^{n},
	\end{equation}
	obtained as a mapping $\wminus\to \wplus$, with $\wminus$, $\wplus$ in $(S^2)^n$, by
	\begin{equation}
		\label{eq:smidpoint}
		\wplus - \wminus = h \xi\paren[\Big]{\frac{\mw_1}{\abs{\mw_1}},\ldots,\frac{\mw_n}{\abs{\mw_n}}}.
	\end{equation}
\end{definition}

\begin{remark}
\label{rk:geodeticmp}
Note that, even for $n=1$, the spherical midpoint method is \emph{not} the geodesic midpoint method on the sphere.
Let $m(\wminus,\wplus)$ denote the geodesic midpoint of $\wminus$ and $\wplus$, and let $d(\wminus,\wplus)$ denote the geodesic (great-circle) distance between $\wminus$ and $\wplus$.
The \emph{geodesic} midpoint method is defined by the conditions that $\xi(m(\wminus,\wplus))$ is tangent to the geodesic between $\wminus$ and $\wplus$, and that ${d(\wminus,\wplus)} = \abs{h\xi(m(\wminus,\wplus))}$.
The spherical midpoint method~\eqref{eq:smidpoint} fulfills the first of these conditions, but not the second: $\xi(m(\wminus,\wplus))$ is tangent to the geodesic between $\wminus$ and $\wplus$, but $2\sin\paren[\big]{d(\wminus,\wplus)/2} = \abs{h\xi(m(\wminus,\wplus))}$.
\end{remark}

Recall now the definition of the classical midpoint method:
\begin{definition}\label{def:classical_midpoint}
	The \textbf{classical midpoint method} for discrete time approximation of the ordinary differential equation $\dot{\w} = X(\w)$, $X\in\Xcal(\RR^{d})$, is the mapping $\wminus \mapsto \wplus$ defined by
	\begin{equation}\label{eq:implicit_midpoint}
		\wplus-\wminus = h X\Big(\frac{\wplus+\wminus}{2}\Big),
	\end{equation}
	where $h>0$ is the time-step length.
\end{definition}

Define a projection map $\projS$ by
\begin{equation}\label{eq:projection_map_various_r}
	\projS(\w) = \paren[\Big]{\frac{\w_1}{\abs{\w_1}},\ldots,\frac{\w_n}{\abs{\w_n}}}.
\end{equation}

It is clear that the spherical midpoint method \eqref{eq:smidpoint} is obtained by 
defining the vector field given by
\begin{equation}\label{eq:relation_ham_vf}
	X(\w) \coloneqq \xi(\projS(\w)).
\end{equation}
and then use the classical midpoint method~\eqref{eq:implicit_midpoint} with the vector field $X$.
Notice that~$X$ is not defined whenever $\w_i=0$ for some~$i$.
In practice this is never a problem, since we are interested in vector fields preserving the spheres.


This indeed gives an integrator on $(S^2)^n$, since the classical midpoint method preserves quadratic invariants, and the vector field~\eqref{eq:relation_ham_vf} is tangent to the spheres (which are the level sets of quadratic functions on $\RR^{3n}$).

We now give the main result of the paper.

\begin{theorem}\label{thm:main}
	The spherical midpoint method \eqref{eq:area_midpoint} fulfils the following properties:
	\begin{enumerate}[label={\textup{(\roman*)}}]
		\item\label{it:symplectic} it is symplectic with respect to $\omega$ if $\xi$ is Hamiltonian with respect to $\omega$;

		\item\label{it:sndorder} it is second order accurate;

		\item\label{it:equi} it is equivariant with respect to $\big(\SO(3)\big)^n$ acting on $(S^{2})^{n}$, i.e.,
		\begin{equation}\label{eq:equivariance_main}
			\psi_{\vect g^{-1}} \circ \Phi(h\xi)\circ \psi_{\vect g} = \Phi(h\psi_{\vect g}^* \xi),
			\quad \forall\, \vect g = (g_1,\ldots,g_n)\in \big( \SO(3) \big)^{n},
		\end{equation}
		where $\psi_{\vect g}$ is the action map;
		

		\item\label{it:linpres} it preserves arbitrary linear symmetries, arbitrary linear integrals, and single-spin homogeneous quadratic integrals $\w_i^\top \vect{A} \w_i$;

		\item\label{it:sadjoint} it is self-adjoint and preserves arbitrary linear time-reversing symmetries;

		\item\label{it:linstable} it is linearly stable: for the linear ODE $\dot{\w} = \lambda \w\times\vect{a}$, 
		the method yields a rotation about the unit vector
		$\vect{a}$ by an angle $\cos^{-1}(1-\frac{1}{2}(\lambda h)^2)$ and hence is stable for $0\le \lambda h<2$. 
	\end{enumerate}
\end{theorem}

\begin{proof}
	We use on several occasions the observation that the spherical midpoint method can be reformulated as the classical midpoint method applied to the vector field \eqref{eq:relation_ham_vf}, using the projection map $\projS$ defined in \eqref{eq:projection_map_various_r}.
	\begin{enumerate}
		\item[\ref{it:symplectic}] The proof is postponed to \autoref{sec:proof_1}.
		\item[\ref{it:sndorder}] The midpoint method is of order~2, and a solution to $\dot{\w} = X(\w)$ with $X$ given by~\eqref{eq:relation_ham_vf} is also a solution to $\dot{\w} = \xi(\w)$.
		\item[\ref{it:equi}] The map $\projS$ is equivariant with respect to $\big(\SO(3)\big)^n$, $\big(\SO(3)\big)^n$ is a subgroup of the affine group on $\RR^{3n}$ and the classical midpoint method is affine equivariant.
		\item[\ref{it:linpres}, \ref{it:sadjoint}] Direct calculations show that $X$ has the same properties in the given cases as the original vector field $\xi$, and the classical midpoint method is known to preserve these properties.
		\item[\ref{it:linstable}] The projection $\projS$ renders the equations for the method nonlinear, even for this linear test equation; it is clear that the solution is a rotation about $\vect{a}$ by \emph{some} angle; this yields a nonlinear equation for the angle with the given solution.
	\end{enumerate}
\end{proof}

\begin{remark}
	Note that the unconditional linear stability of the classical midpoint method is lost for the spherical midpoint method; the method's response to the harmonic oscillator is identical to that of the leapfrog (Störmer--Verlet) method.
\end{remark}

\begin{remark}
	The spherical midpoint method is second order accurate.
	Since it is also symmetric, one can use symmetric composition techniques, as described in~\cite[\S\!~V.3.2]{HaLuWa2006}, to obtain higher order symplectic integrators on $(S^2)^n$.
\end{remark}


\subsection{Spherical Midpoint Method: Lie--Poisson integrator} 
\label{sub:interpretation_as_lie_poisson_integrator}

$\RR^{3n}$ is a Lie--Poisson manifold with Poisson bracket
\begin{equation}\label{eq:LiePoisson_bracket}
	\{F,G \}(\w) = \sum_{k=1}^{n}\paren[\Big]{
		\frac{\pd F(\w)}{\pd \w_k}\times 
		\frac{\pd G(\w)}{\pd \w_k}
	}\cdot \w_k.
\end{equation}
This is the canonical Lie--Poisson structure of $(\so(3)^{*})^{n}$, or $(\su(2)^{*})^{n}$, obtained by identifying $\so(3)^{*}\simeq \RR^{3}$, or $\su(2)^{*}\simeq \RR^{3}$.
For details, see \cite[\S\!~10.7]{MaRa1999} or \cite{McMoVe2014c}.

The Hamiltonian vector field associated with a Hamiltonian function $H\colon \RR^{3n}\to \RR$ is given by
\begin{equation}\label{eq:Ham_vf}
	X_{H}(\w) = \sum_{k=1}^{n} \w_k\times \frac{\pd H(\w)}{\pd \w_k} .
\end{equation}
Its flow, $\exp(X_{H})$, preserves the Lie--Poisson structure, i.e.,
\begin{equation}\label{eq:preservation_of_LP}
	\{F\circ\exp(X_{H}),G\circ\exp(X_{H}) \} = \{F,G \}\circ\exp(X_{H}), \quad \forall F,G\in C^{\infty}(\RR^{3n}).
\end{equation}
The flow $\exp(X_{H})$ also preserves the \emph{coadjoint orbits}~\cite[\S\!~14]{MaRa1999}, given by
\begin{equation}\label{eq:coadjoint_orbits}
	S_{\lambda_1}^{2}\times\cdots\times S_{\lambda_n}^2\subset \RR^{3n},\quad \lambda_1,\ldots,\lambda_n \geq 0,  
\end{equation}
where $S^{2}_{\lambda}$ denotes the 2--sphere in $\RR^3$ of radius~$\lambda$.
A \emph{Lie--Poisson integrator} for $X_{H}$ is an integrator that, like the exact flow, preserves the Lie--Poisson structure and the coadjoint orbits.
For an illustration of the coadjoint orbits, see \autoref{fig:coadjoint}.

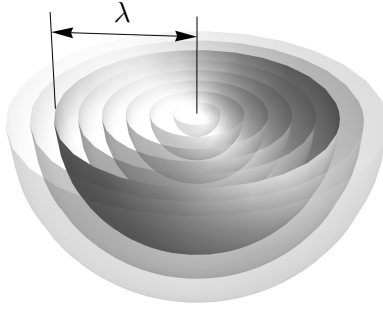
\begin{figure}
	\centering
	\input{fig_symplectic_leaves}
	\caption{
	Structure of the Lie--Poisson manifold $(\RR^3,\{\cdot,\cdot\})$.
	Lie--Poisson manifolds are foliated by symplectic submanifolds \emph{(symplectic leaves)} given by the coadjoint orbits.
	For $\RR^3$ equipped with the Poisson bracket~(\ref{eq:LiePoisson_bracket}), the coadjoint orbits are given by the submanifolds $S^2_{\lambda}\subset \RR^{3}$.
	Thus, to construct a Lie--Poisson integrator on $\RR^{3n}$ is equivalent to constructing symplectic integrators for the symplectic direct product manifolds $S_{\lambda_1}^{2}\times\cdots\times S_{\lambda_n}^2$.}
	\label{fig:coadjoint}
\end{figure}

\begin{definition}\label{def:extended_spherical_midpoint}
	The \textbf{extended spherical midpoint method} for $X\in\Xcal(\RR^{3n})$ is the numerical integrator defined by
	\begin{equation}\label{eq:area_midpoint_LP}
		\wplus-\wminus = h X
	\paren[\bigg]{ \frac{\sqrt{\abs{\wminus_{1}}\abs{\wplus_{1}}}(\wminus_{1}+\wplus_{1})}{\abs{\wminus_{1}+\wplus_{1}} },\ldots, \frac{\sqrt{\abs{\wminus_{n}}\abs{\wplus_{n}}}(\wminus_{n}+\wplus_{n})}{\abs{\wminus_{n}+\wplus_{n}} } }
	.
	\end{equation}
\end{definition}
We define the expression $\frac{\sqrt{\abs{\wminus_{i}}\abs{\wplus_{i}}}(\wminus_{i}+\wplus_{i})}{\abs{\wminus_{i}+\wplus_{i}} }$ to be zero whenever the denominator is zero.
The equation~\eqref{eq:area_midpoint_LP} is thereby defined on all of $\RR^{3n}$.


We have the following result, analogous to~\autoref{thm:main}.


\begin{theorem}\label{thm:mainLP}
	The extended spherical midpoint method \eqref{eq:area_midpoint_LP} fulfils the following properties:
	\begin{enumerate}[label={\textup{(\roman*)}}]
		\item\label{it:liepoisson} it is a Lie--Poisson integrator for Hamiltonian vector fields $X_{H}\in \Xcal(\RR^{3n})$;
		\item\label{it:sndorder2} it is second order accurate;

		\item\label{it:equi2} it is equivariant with respect to $\big(\SO(3)\big)^n$ acting diagonally on $(\RR^{3})^n\simeq \RR^{3d}$ (the diagonal action is defined by $(g_1,\ldots,g_n) \cdot (\w_1,\ldots,\w_n) = (g_1\w_1,\ldots,g_n\w_n)$).
		

		\item\label{it:linpres2} it preserves arbitrary linear symmetries, arbitrary linear integrals, and single-spin homogeneous quadratic integrals $\w_i^\top \vect{A} \w_i$, where $\vect A\in \R^{3\times 3}$;

		\item\label{it:sadjoint2} it is self-adjoint and preserves arbitrary linear time-reversing symmetries;

	\end{enumerate}

\end{theorem}

\begin{proof}
	For convenience, we define
 $\Gamma\colon \RR^{3n}\times\RR^{3n}\to \RR^{3n}$ by
\begin{equation}\label{eq:Gamma_projection}
	\Gamma\big(\wminus,\wplus\big) \coloneqq
	\paren[\bigg]{ \frac{\sqrt{\abs{\wminus_{1}}\abs{\wplus_{1}}}(\wminus_{1}+\wplus_{1})}{\abs{\wminus_{1}+\wplus_{1}} },\ldots, \frac{\sqrt{\abs{\wminus_{n}}\abs{\wplus_{n}}}(\wminus_{n}+\wplus_{n})}{\abs{\wminus_{n}+\wplus_{n}} } }.
\end{equation}
	\begin{enumerate}
		\item[\ref{it:liepoisson}] The proof is postponed to \autoref{sec:proof_1}.

		\item[\ref{it:sndorder2}] First notice that
		\begin{equation}\label{eq:proof_Gamma_midpoint}
			\Gamma(\wminus,\wplus) = \frac{\wminus+\wplus}{2} + \mathcal{O}(\abs{\wplus-\wminus})
		\end{equation}
		Using \eqref{eq:proof_Gamma_midpoint} in~\eqref{eq:area_midpoint_LP}, and using that $X$ is smooth, we obtain
		\begin{equation}
			\wplus-\wminus = h X\Big(\frac{\wminus+\wplus}{2} \Big) + h\mathcal{O}(\abs{\wplus-\wminus}).
		\end{equation}
		We use~\eqref{eq:area_midpoint_LP} again to obtain
		\begin{equation}
			\wplus-\wminus = h X\Big(\frac{\wminus+\wplus}{2} \Big) + h^{2}\mathcal{O}\big(\big| X\big( \Gamma(\wminus,\wplus) \big)  \big|\big).
		\end{equation}
		Since $\Gamma(\wminus,\wplus)$ is bounded for fixed $\wminus$, we get $\wplus=\tilde{\wplus} + \mathcal{O}(h^{2})$, where $\tilde{\wplus}$ is the solution obtained by the classical midpoint method~\eqref{eq:implicit_midpoint} on $\RR^{3n}$.
		The method defined by~\eqref{eq:area_midpoint_LP} is therefore at least first order accurate.
		Second order accuracy follows since the method is symmetric.

		\item[\ref{it:equi2}]  $\Gamma$ is equivariant with respect to $(\SO(3))^n$, so we obtain $\SO(3)$ equivariance of the method.

		\item[\ref{it:linpres2}, \ref{it:sadjoint2}] Same proof as in \autoref{thm:main}.
	\end{enumerate}
\end{proof}

\subsection{Proof of symplecticity} 
\label{sec:proof_1}

We need some preliminary definitions and results before the main proof.

\begin{definition}\label{def:rays}
	The \textbf{ray} through a point $\w\in \RR^{3n}$ is the subset 
	\begin{equation}
		\{ (\lambda_1\w_1,\ldots,\lambda_n\w_n); \vect\lambda\in \RR_+^{n} \}.
	\end{equation}
\end{definition}

The set of all rays is in one-to-one relation with $(S^2)^n$.
Note that the vector field~$X$ defined by~\eqref{eq:relation_ham_vf} is \emph{constant on rays}.
The following result, essential throughout the remainder of the paper, shows that the property of being constant on rays is passed on from Hamiltonian functions to Hamiltonian vector fields.

\begin{lemma}
	\label{lma:XHconstant}
	If a Hamiltonian function $H\in C^{\infty}((\RR^{3}\backslash\{0\})^n)$ is constant on rays, then so is its Hamiltonian vector field $X_H$.
\end{lemma}

\begin{proof}
	It is enough to consider $n=1$, as the general case proceeds the same way.
	$H$ is constant on rays, so for $\lambda>0$, we have
	\begin{align}
		H(\lambda\w) = H(\w).
	\end{align}
	Differentiating with respect to $\bw$ yields
	\begin{align}
		\lambda \nabla H(\lambda \w) = \nabla H (\w)
	\end{align}
	The Hamiltonian vector field at $\lambda \bw$ is
	\begin{equation}
		\begin{split}
			X_H(\lambda \w) &= \lambda \w \times \nabla H(\lambda \w) \\
			&= \w \times \nabla H(\w)\\
			&= X_H(\w),	
		\end{split}
	\end{equation}
	which proves the result.
\end{proof}

Recall that if $X$ is any vector field on $\RR^n$, then tangent vectors $\u (t)$ to integral curves $\w(t)$
of $X$ obey the variational equation $\dot \u  = DX(\w(t))\u $, where $\u \in T_{\w(T)}\RR^n$. The following
lemma establishes the equivalent result for transport of 1-forms. We represent the
1-form $\sum_{i=1}^n \sigma_i dw_i\in T_{\w}^*\RR^n$ by the column vector $\sigma$. 

\begin{lemma}
\label{lma:1forms}
Let  $\varphi(t)$ the flow of the vector field $X$ on $\RR^n$ and $\w(t)$ an integral curve. Let $\sigma(t)$ be
a curve of 1-forms transported by the flow, i.e., such that $\varphi(t)^*\sigma(t) = \sigma(0)$. Then $\dot\sigma = -DX(\w(t))^\top\sigma$.
\end{lemma}

\begin{proof}
For all $\u \in T_{\w(0)}\RR^n$ we have
$
\langle \varphi(t)^* \sigma(t),\u \rangle = \langle \sigma(t),D\varphi(t) \u \rangle,
$
so that $\sigma(0)^\top \u  = \sigma(t)^\top D\varphi(t) \u $ or $\sigma(0) = D\varphi(t)^\top \sigma(t)$.
Differentiating with respect to $t$ at $t=0$ and using $D\varphi(0) = I$, $\dot\varphi(0) = X$
gives the result.
\end{proof}

Any Poisson bracket on a manifold $M$ is associated with a Poisson bivector $K$, a section of $\bigwedge^2(TM)$, such that $\{F,G\}(\w) = K(\w)\paren[\big]{{\mathrm{d}}F(\w),{\mathrm{d}}G(\w)}$. The flow of a Hamiltonian vector field preserves the Poisson structure (see, e.g., \cite{MaRa1999}, Prop. 10.3.1), which in terms of $K$ is the statement
that $\frac{d}{dt}K(\w(t)(\sigma(t),\lambda(t))$=0. In the Lie--Poisson case, $K$ is linear in $\w$, so using the product rule together with linearity in each of the 3 arguments gives
\begin{equation}
\label{eq:3terms}
K(\dot \w)(\sigma,\lambda) + K(\w)(\dot\sigma,\lambda) + K(\w)(\sigma,\dot\lambda) = 0
\end{equation}
where $\dot\w = X_H(\w)$ and from Lemma \ref{lma:1forms}, $\dot\sigma = - (DX_H)^\top \sigma$ and $\dot\lambda = -(DX_H)^\top\lambda$.

\newcommand\liealg{\mathfrak{g}}
\newcommand\Lie[2]{\mathcal{L}_{#1}{#2}}

\begin{lemma}\label{lem:poisson_midpoint}
	Let $H\in C^\infty((\RR^{3}\backslash\{0\})^n)$ be constant on rays, and let $X :=X_H$ denote its Hamiltonian vector field.
	Then the classical midpoint method (\autoref{def:classical_midpoint}) applied to $X$ is a Lie--Poisson integrator. 
\end{lemma}

\begin{proof}
From \autoref{lma:XHconstant}, the Hamiltonian vector field $X$ is constant on rays. 

In addition, $X$ is tangent to the coadjoint orbits, which are the level sets of the quadratics $|\vect w_1|^2,\dots,|\vect w_n|^2$, so the classical midpoint method applied to $X$ preserves the coadjoint orbits. 
We will show that it is also a Poisson map with respect to the Poisson bracket~\eqref{eq:LiePoisson_bracket}.

In terms of the Poisson bivector $K$, to establish that a map $\varphi\colon\wminus\mapsto \wplus$ is Poisson is equivalent to showing that $K$ is preserved, i.e., that $K(\wplus)(\aplus,\bplus) = K(\wminus)(\aminus,\bminus)$ for all 1-forms $\aplus,\bplus\in T_{\wplus}^* M$, where $\aminus=\varphi^* \aplus$ and $\bminus=\varphi^* \bplus$.
Let $\wmid \coloneqq (\wminus+\wplus)/2$ and $\wdif \coloneqq \wplus-\wminus$. Then the classical midpoint
method applied to $X$ takes the form $\wdif = h X(\wmid)$. 
Therefore, introducing
$\adif \coloneqq \aplus - \aminus$ and $ \amid \coloneqq \frac{1}{2}(\aminus+\aplus)$, we have
$\adif =  - h DX(\wmid)^\top \amid$ and 
similarly 
$\bdif \coloneqq \bplus - \bminus$, $ \bmid \coloneqq \frac{1}{2}(\bminus+\bplus)$,
and
$\bdif =  - h DX(\wmid)^\top \bmid$.

%

In the Lie--Poisson case~\eqref{eq:LiePoisson_bracket}, $K(\w)$ is linear in $\w$ and so linearity in all three arguments gives after cancellations:
\begin{multline}
	K(\wplus)(\aplus,\bplus) - K(\wminus)(\aminus,\bminus) =\\
	\underbrace{K(\wdif)(\adif,\bdif)}_{\Delta_1} 
	+ \underbrace{K(\wdif)(\amid,\bmid)+K(\wmid)(\adif,\bmid)+K(\wmid)(\amid,\bdif)}_{\Delta_2}.
\end{multline}

The term $\Delta_2$ vanishes because the 3 terms are precisely those appearing in \eqref{eq:3terms}.
(In fact, $\Delta_2=0$ for the classical midpoint method applied to {\em any} Lie--Poisson system, essentially because $\wdif$ is a Poisson vector field evaluated at $\wmid$.)

We now look at the term $\Delta_1$.
For the Poisson structure~\eqref{eq:LiePoisson_bracket}, $K(\w)(\aneutral,\bneutral) = \sum_{i=1}^n\det([\vect w_i,\aneutral_i,\bneutral_i])$. 
Therefore
\begin{align}
	K(\wdif)(\adif,\bdif) 
	&= \sum_{i=1}^n \det([\wdif_i,\adif_i,\bdif_i]) \\
	&= h^3 \sum_{i=1}^n \det([\f(\wmid)_i, (-D\f(\wmid)^{\top}\amid)_i,(-D\f(\wmid)^{\top}\bmid)_i])\\
	&=0 
\end{align}
because $\f(\wmid)_i$, $(-D\f(\wmid)^{\top}\amid)_i$, and $(-D\f(\wmid)^{\top}\bmid)_i$ are all 
orthogonal to $\vect w_i$: $\f(\wmid)_i$, because it is tangent to the 2-spheres, and $(-D\f(\wmid)^{\top}\amid)_i$ and $(-D\f(\wmid)^{\top}\bmid)_i$, because $\langle \wmid_i, (-D\f(\wmid)^{\top}\amid)_i\rangle = -\langle (D\f(\wmid)\wmid)_i,\amid_i\rangle$, which is zero because $\f$ is constant on rays.
We have shown that the classical midpoint method applied to $X$ is Poisson and preserves the symplectic leaves, thus it is symplectic on them. 
This establishes the result.
\end{proof}

\begin{proof}[Proof of \autoref{thm:main}-\ref{it:symplectic}.]
	The symplectic form $\tilde\omega$ on $S^2_{\lambda_1}\times\cdots\times S^2_{\lambda_n}$ induced by the Lie--Poisson structure on $\RR^{3n}$ is given by
	\begin{equation}
		\tilde\omega_{\w}(\vect{u},\vect{v}) = \sum_{i=1}^n \vect{u}_i\times\vect{v}_i\cdot \w_i.
	\end{equation}
	Likewise, the symplectic structure $\omega$ on $(S^2)^n$ given by~\eqref{eq:symp_form} can be written
	\begin{equation}
		\omega_{\w}(\vect{u},\vect{v}) = \sum_{i=1}^n \kappa_i\vect{u}_i\times\vect{v}_i\cdot \w_i.
	\end{equation}
	The mapping $\Phi\colon((S^2)^n,\omega)\to (S^2_{\kappa_1}\times\cdots\times S^2_{\kappa_n},\tilde\omega)$ given by $\w_i \mapsto \kappa_i \w_i$ is therefore a symplectomorphism (a symplectic diffeomorphism).
	Thus, the spherical midpoint method~\eqref{eq:area_midpoint} is symplectic on $((S^2)^n,\omega)$ if and only if it is symplectic on $(S^2_{\kappa_1}\times\cdots\times S^2_{\kappa_n},\tilde\omega)$ when represented in the variables $\tilde{\wminus} = \Phi(\wminus)$ and $\tilde{\wplus}=\Phi(\wplus)$.
	Let $H$ be the Hamiltonian function corresponding to a Hamiltonian vector field $\xi$ on $(S^2)^n$.
	Let $\bar H\in C^{\infty}((\RR^3\backslash\{0\})^n)$ be the extension to a ray-constant Hamiltonian. 
	A short calculation shows that the spherical midpoint method~\eqref{eq:smidpoint} for the Hamiltonian vector field $\xi$, but expressed in the variables $\tilde{\wminus}$ and $\tilde{\wplus}$, can be written
	\begin{equation}\label{eq:tildeeq}
		\tilde{\wplus}-\tilde{\wminus} = h X_{\bar H}\left(\rho\left(\frac{\tilde{\wplus}+\tilde{\wminus}}{2}\right)\right).
	\end{equation}
	Since $\bar H$ is constant on rays, it follows from~\autoref{lma:XHconstant} that $X_{\bar H}$ is constant of rays.
	Therefore, $X_{\bar H}\circ\rho = X_{\bar H}$.
	It follows follows from~\autoref{lem:poisson_midpoint} that $\tilde{\wminus}\mapsto \tilde{\wplus}$ defined by~\eqref{eq:tildeeq} is a symplectic mapping with respect to $\tilde\omega$.
	This proves the result.
\end{proof}

\begin{proof}[Proof of \autoref{thm:mainLP}-\ref{it:liepoisson}.]
	We need to prove that the method $\wminus\mapsto\wplus$ defined by~\eqref{eq:area_midpoint_LP} with $X=X_H$ is a Lie--Poisson map that preserves the coadjoint orbits.
	Equivalent is to prove that if $\wminus \in S^2_{\lambda_1}\times\cdots\times S^2_{\lambda_n}$, with $\lambda_i\geq 0$, then $\wminus\mapsto\wplus$ is a symplectic mapping $S^2_{\lambda_1}\times\cdots\times S^2_{\lambda_n}\to S^2_{\lambda_1}\times\cdots\times S^2_{\lambda_n}$ (with respect to the symplectic structure on $S^2_{\lambda_1}\times\cdots\times S^2_{\lambda_n}$ induced by the Lie--Poisson structure of $\RR^3$).
	If $\lambda_k=0$ for some $k$, i.e., $\wminus_k=0$, then $X_H(\wminus)_k = 0$ and if follows from~\eqref{eq:area_midpoint_LP} that $\wplus_k=0$.
	Thus, the variables $\wminus_k$ and $\wplus_k$ are constants that do not affect the dynamics (they can be removed from phase space).
	It is therefore no restriction to assume that $\lambda_i>0$ for all~$i$.
	Now define a Hamiltonian function $\bar H$ on $(\RR^3\backslash\{0\})^n$ by extending $H|_{S^2_{\lambda_1}\times\cdots\times S^2_{\lambda_n}}$ to be constant on the rays.
	By~\autoref{lem:poisson_midpoint}, the classical midpoint method applied to $X_{\bar H}$ is a Lie--Poisson integrator.
	In particular, it defines a symplectic map $\varphi_h\colon S^2_{\lambda_1}\times\cdots\times S^2_{\lambda_n}\to S^2_{\lambda_1}\times\cdots\times S^2_{\lambda_n}$.
	If $\wplus \coloneqq \varphi_h(\wminus)$, then $\wminus$ and $\wplus$ fulfill equation~\eqref{eq:area_midpoint_LP} with $X=X_H$, since $\abs{\wminus_i}=\abs{\wplus_i}=\lambda_i$ and $X_H|_{S^2_{\lambda_1}\times\cdots\times S^2_{\lambda_n}} = X_{\bar H}|_{S^2_{\lambda_1}\times\cdots\times S^2_{\lambda_n}}$.
	This proves the result.
\end{proof}

\section{Examples} 
\label{sec:example}

\subsection{Single particle system: free rigid body} 
\label{sub:single_particle_system_standard_rigid_body}

\begin{figure}
	\centering
	\includegraphics[width=.45\textwidth]{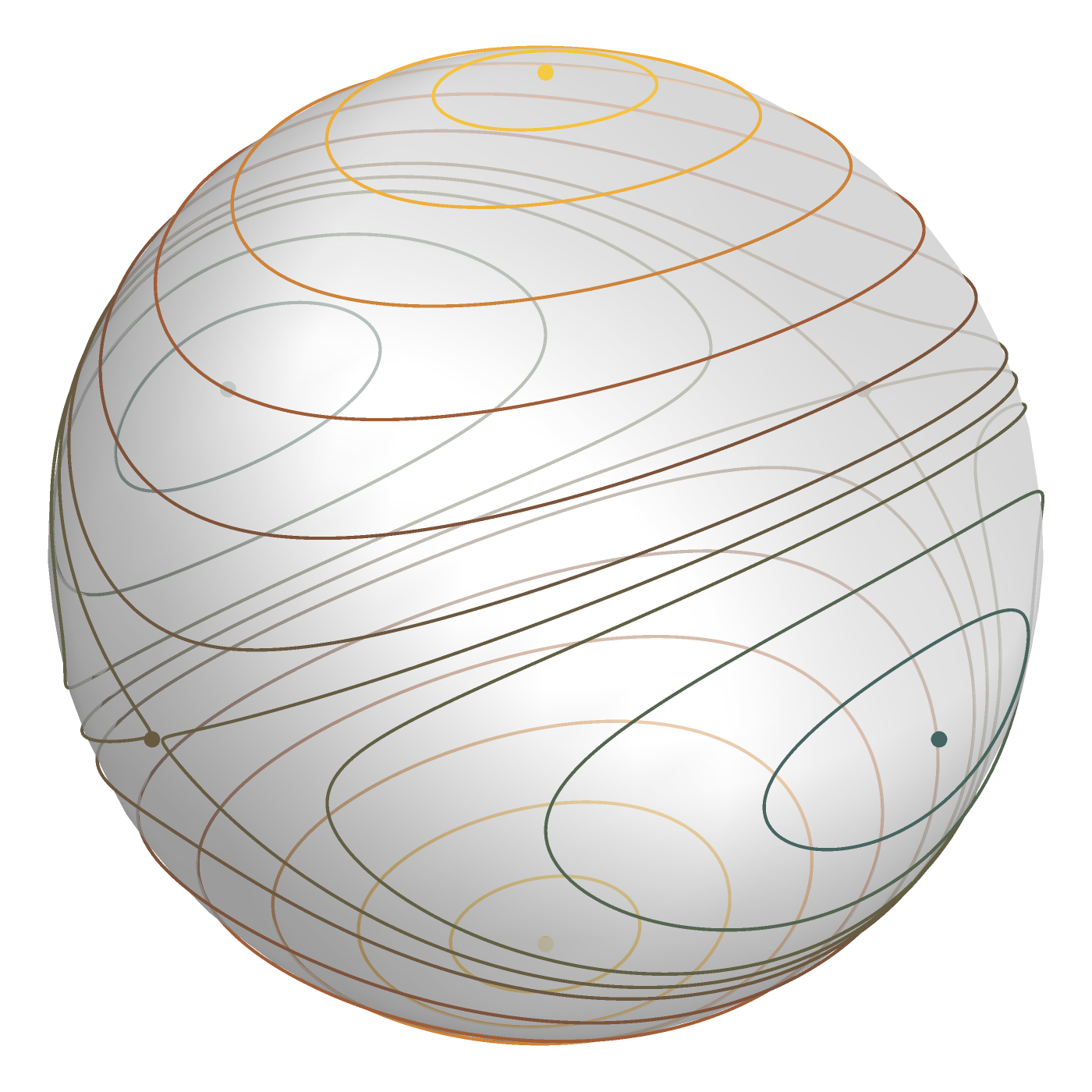}
	\caption{
		Phase portrait for the free rigid body problem with Hamiltonian~(\ref{eq:Ham_symrb}).
		The system has relative equilibria at the poles of the principal axes.
		The phase portrait is invariant under central inversions due to time-reversal symmetry 
		$H(\vect w) = H(-\vect w)$ of the Hamiltonian.
	}
	\label{fig_symrb_phaseplot}
\end{figure}

\begin{figure}
	\input{fig_symrb_error_plot}
	\caption{
		Errors $\max_{k} \abs{\vect{w}_k-\vect{w}(h k)}$  at different time-step lengths~$h$, for three different approximations of the free rigid body.
		The time interval is $0\leq hk \leq 10$ and the initial data are $\vect w_0 = (\cos(1.1),0,\sin(1.1))$.
		The errors for the spherical midpoint method are about 400 times smaller than the corresponding errors for the discrete Moser--Veselov algorithm and about 30 times smaller than the corresponding errors for the classical midpoint method.
	}
	\label{fig_symrb_error_plot}
\end{figure}
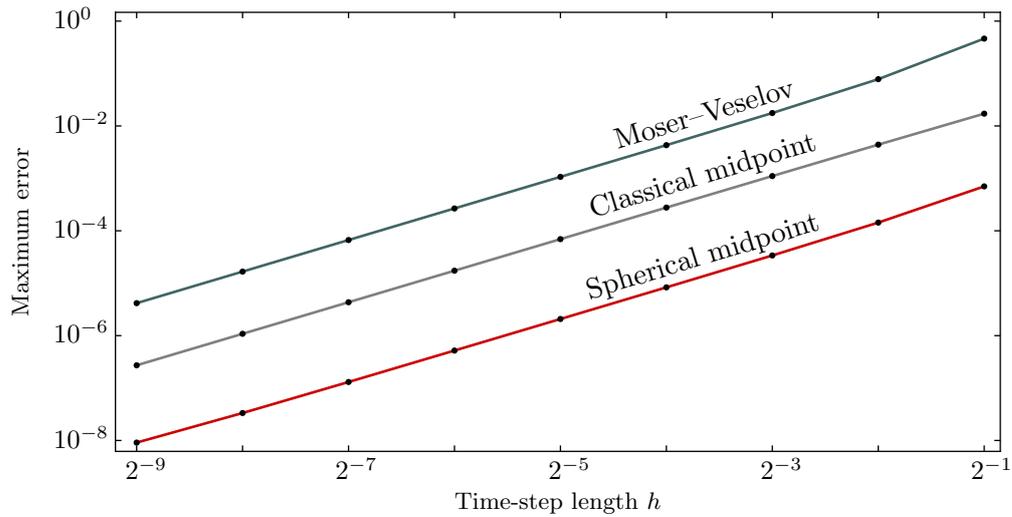
Consider a single particle system on $S^2$ with Hamiltonian
\begin{equation}\label{eq:Ham_symrb}
	H(\w) = \frac{1}{2}\w\cdot \vect{I}^{-1}\w,
\end{equation}
where $\vect{I}$ is an inertia tensor, given by
\begin{equation}
	\vect{I} = \begin{pmatrix}
		I_1 & 0 & 0 \\
		0 & I_2 & 0 \\
		0 & 0 & I_3 \\
	\end{pmatrix}
	,\quad
	I_1 = 1,\; I_2 = 2,\; I_3 = 4.
\end{equation}
This system describes a free rigid body.
Its phase portrait is given in \autoref{fig_symrb_phaseplot}.
The poles of the principal axes are relative equilibria, and every trajectory is periodic (as expected for 2--dimensional Hamiltonian systems).
Also note the time-reversal symmetry $\w\mapsto -\w$.

We consider three different discrete approximations: the discrete Moser--Veslov algorithm~\cite{MoVe1991}, the classical midpoint method~\eqref{eq:implicit_midpoint}, and the spherical midpoint method~\eqref{eq:area_midpoint}.
All these methods exactly preserve the Hamiltonian~\eqref{eq:Ham_symrb}, so each discrete trajectory lies on a single trajectory of the continuous system: if $\w_0,\w_1,\w_2,\ldots$ is a discrete trajectory, and $\w(t)$ is the continuous trajectory that fulfils $\w(0) = \w_0$, then $\w_k\in \w(\RR)$.
There are, however, phase errors: if $\w_0,\w_1,\w_2,\ldots$ is a discrete trajectory with time-step length $h$, and $\w(t)$ is the continuous trajectory that fulfils $\w(0) = \w_0$, then  $e_k \coloneqq \abs{\w_k-\w(hk)} \neq 0$ (in general).
The maximum error in the time interval $t\in [0,10]$ for the three methods, with initial data $\vect w_0 = (\cos(1.1),0,\sin(1.1))$ and various time-step lengths, is given in \autoref{fig_symrb_error_plot}.
The spherical midpoint method produce errors about 400 times smaller than errors for the discrete Moser--Veselov algorithm, and about 30 times smaller than errors for the classical midpoint method.

The discrete model of the free rigid body obtained by the spherical midpoint discretisation is \emph{discrete integrable} (c.f.~\cite{MoVe1991}), i.e., it is a symplectic mapping $S^2\to S^2$ with an invariant function (or, equivalently, it is a Poisson mapping $\RR^3\to \RR^3$ with two invariant functions that are in involution).
An interesting future topic is to attempt to generalise this integrable mapping to higher dimensions, and to characterise its integrability in terms of Lax pairs.
For the Moser--Veselov algorithm, such studies have led to a rich mathematical theory~\cite{DeLiTo1992}.

\subsection{Single particle system: irreversible rigid body} 
\label{sub:1_particle}


Consider a single particle system on $S^2$ with Hamiltonian
\begin{equation}\label{eq:Ham_nonsymrb}
	H(\w) = \frac{1}{2}\w\cdot \vect{I}(\w)^{-1}\w,
\end{equation}
where $\vect{I}(\w)$ is an irreversible inertia tensor, given by
\begin{equation}\label{eq:noniso_inertia}
	\vect{I}(\w) = \begin{pmatrix}
		\frac{I_1}{1 + \sigma w_1} & 0 & 0 \\
		0 & \frac{I_2}{1 + \sigma w_2} & 0 \\
		0 & 0 & \frac{I_3}{1 + \sigma w_3} \\
	\end{pmatrix}
	,\quad
	I_1 = 1,\; I_2 = 2,\; I_3 = 4,\; \sigma = \frac{2}{3}.
\end{equation}
This system describes an irreversible rigid body with fixed unitary total angular momentum.
It is irreversible in the sense that the moments of inertia about the principal axes depend on the rotation direction, i.e., the moments for clockwise and anti-clockwise rotations are different.
A phase portrait is given in \autoref{fig_nonsymrb}\subref{fig_nonsymrb_phaseplot}.
Like the free rigid body, the poles of the principal axes are relative equilibria, and every trajectory is periodic.
Contrary to the free rigid body, the phase portrait is not symmetric under central inversions, i.e., there is no apparent time-reversal symmetry.

We consider two different discrete approximations: the classical midpoint method~\eqref{eq:implicit_midpoint} and the spherical midpoint method~\eqref{eq:area_midpoint}.
Locally the two methods are akin (they are both second order accurate), but they exhibit distinct global properties: trajectories lie on periodic curves for the spherical midpoint method but not for the classical midpoint method; see~\autoref{fig_nonsymrb}\subref{fig_nonsymrb_traj_drift}.
Also, the deviation in the Hamiltonian~\eqref{eq:Ham_nonsymrb} along discrete trajectories remains bounded for the spherical midpoint method, but  drifts for the classical midpoint method; see~\autoref{fig_nonsymrb}\subref{fig_nonsymrb_energy_drift}.

Periodicity of phase trajectories and near conservation of energy, as displayed for the spherical midpoint method, suggests the presence of a first integral, a \emph{modified Hamiltonian}, that is exactly preserved.
The existence of such a modified Hamiltonian hinges on symplecticity, as established through the theory of
\emph{backward error analysis}~\cite{HaLuWa2006}.

The example in this section illustrates the advantage of the spherical midpoint method, over the classical midpoint method, for approximating Hamiltonian dynamics on~$S^2$.
In general, one can expect that spherical midpoint discretisations of continuous integrable systems on $(S^{2})^{n}$ remain \emph{almost integrable} in the sense of Kolmogorov--Arnold--Moser theory for symplectic maps, as developed by \citet{Sh1999b}.


\begin{figure}
	\centering
	\subfloat[Phase portrait]{\label{fig_nonsymrb_phaseplot}\includegraphics[width=.45\textwidth]{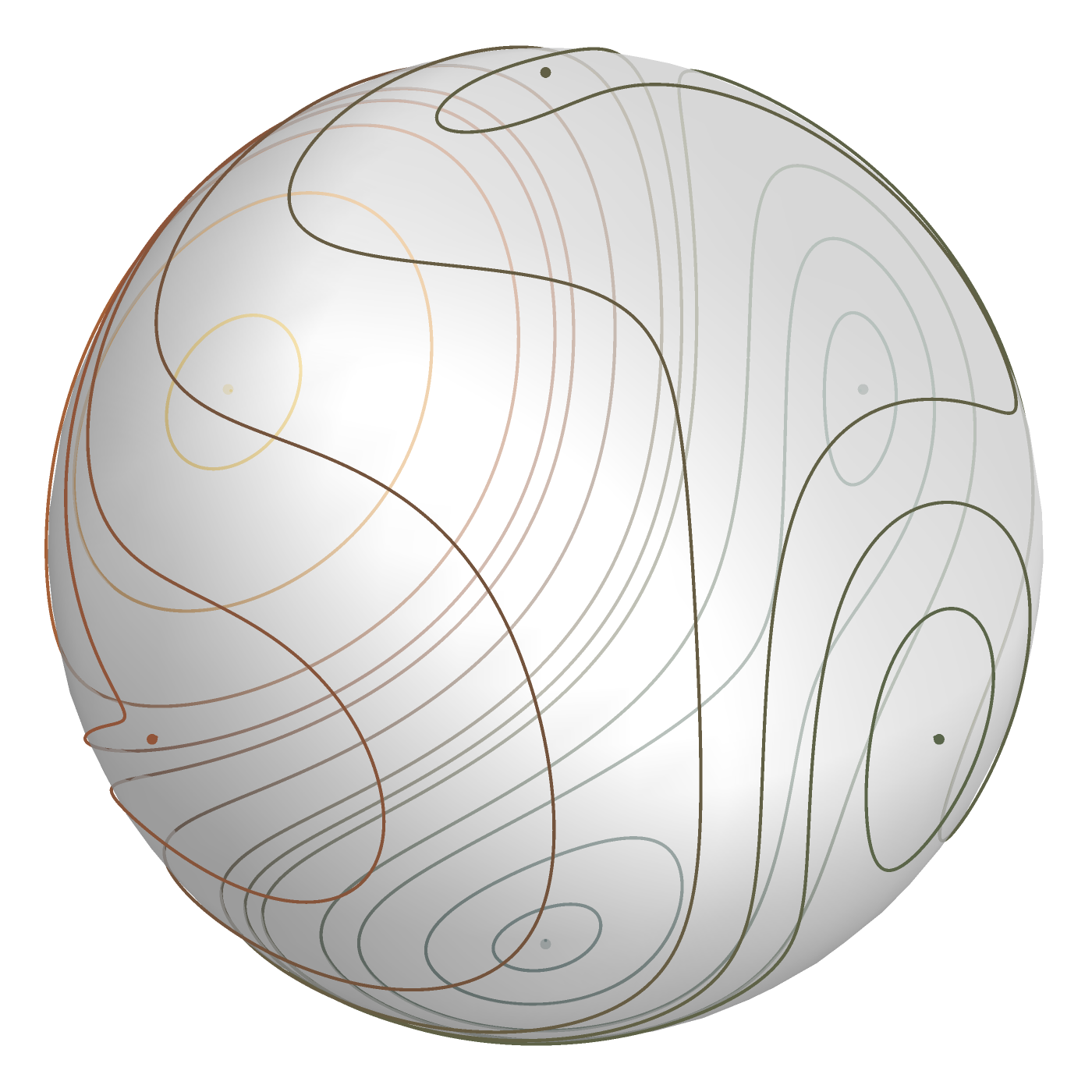}}
	\quad\quad
	\subfloat[Computed trajectories]{\label{fig_nonsymrb_traj_drift}\includegraphics[width=.45\textwidth]{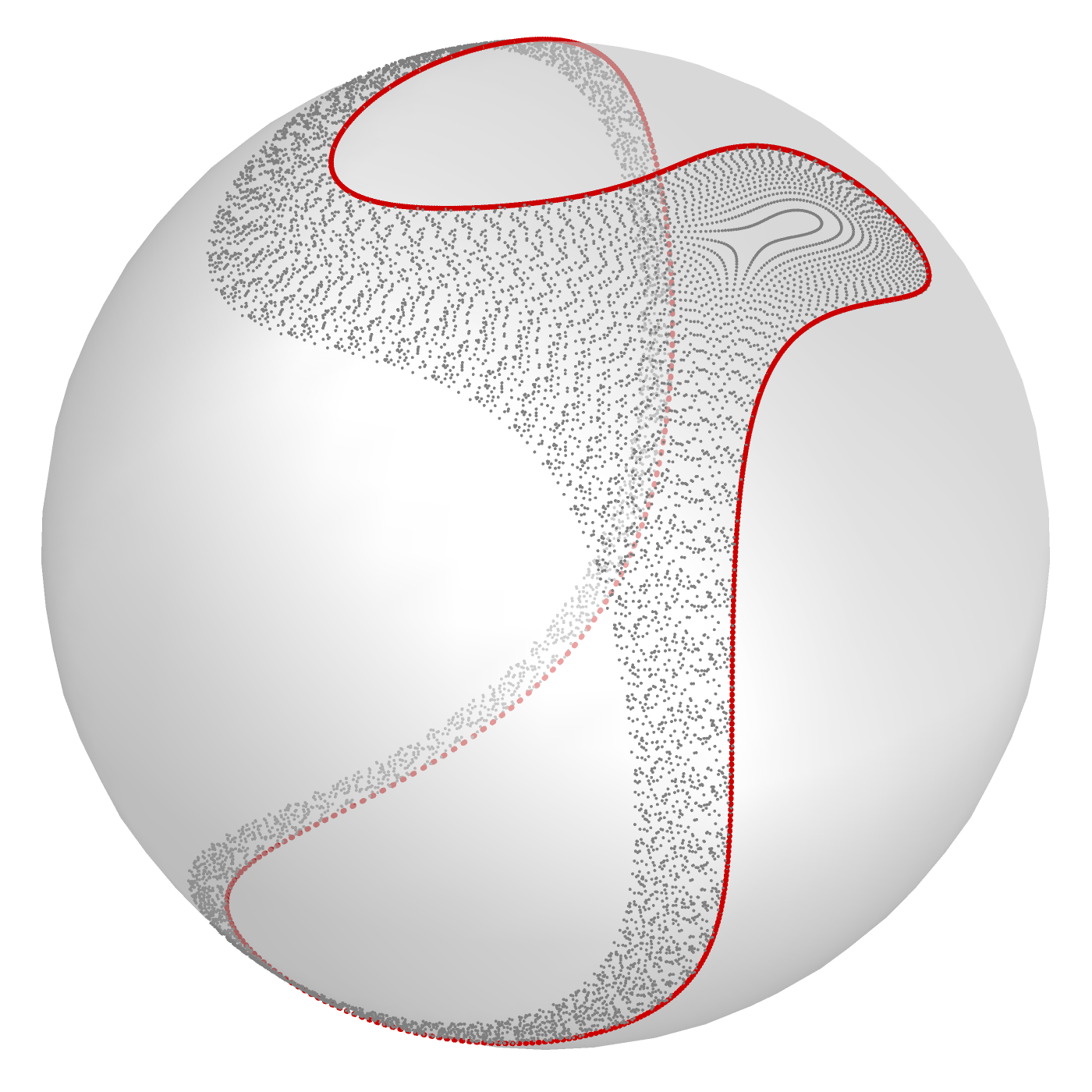}} \\
	\subfloat[Energy errors]{\label{fig_nonsymrb_energy_drift}\input{fig_nonsymrb_energy_drift}}
	\caption{
		The irreversible rigid body problem with Hamiltonian~(\ref{eq:Ham_nonsymrb}) is approximated by the classical midpoint method and the spherical midpoint method.
		The time step used is $h=1/2$.
		\protect\subref{fig_nonsymrb_phaseplot} Phase portrait obtained using the spherical midpoint method.
		The system has the same relative equilibria as the classical rigid body, but has no affine reversing symmetry.
		\protect\subref{fig_nonsymrb_traj_drift} Two corresponding discrete trajectories: the classical midpoint method (grey) and the spherical midpoint method (red).
		The initial data are $\vect w_0 = (0,0.7248,-0.6889)$.
		The trajectory obtained with the spherical midpoint lies on a smooth closed curve. 
		\protect\subref{fig_nonsymrb_energy_drift} Energy error $H(\vect{w}_k)-H(\vect{w}_0)$ for a two discrete trajectories.
		The energy drifts for the classical midpoint method, but remains bounded for the spherical midpoint method.
	}
	\label{fig_nonsymrb}
\end{figure}


\subsection{Single particle system: forced rigid body, development of chaos} 
\label{sub:forced_1_particle_system_chaotic_behaviour}

Consider the time dependent Hamiltonian on $S^2$ given by
\begin{equation}\label{eq:forced_rb_Ham}
	H(\w,t) = \frac{1}{2}\w\cdot \vect{I}^{-1}\w + \varepsilon\sin(t)w_3,\quad \w = (w_1,w_2,w_3),
\end{equation}
where $\vect{I}$ is an inertia tensor, given by
\begin{equation}
	\vect{I} = \begin{pmatrix}
		I_1 & 0 & 0 \\
		0 & I_2 & 0 \\
		0 & 0 & I_3 \\
	\end{pmatrix}
	,\quad
	I_1 = 1,\; I_2 = 4/3,\; I_3 = 2.
\end{equation}
This system describes a forced rigid body with periodic loading of period $2\pi$.
At $\varepsilon=0$ the system is integrable, but it becomes non-integrable as $\varepsilon$ increases.
We discretise the system using the spherical midpoint method with time-step length $2\pi/N$, $N=20$.
A Poincaré section is obtain by sampling the system every $N$:th step; the result for various initial data and choices of $\varepsilon$ is shown in \autoref{forced_rb_fig}.
Notice the development of chaotic behaviour near the unstable equilibria points.

The example in this section illustrates that the spherical midpoint method, being symplectic, behaves as expected in the transition from integrable to chaotic dynamics.

\begin{figure}
	\centering
	\includegraphics[width=.4\textwidth]{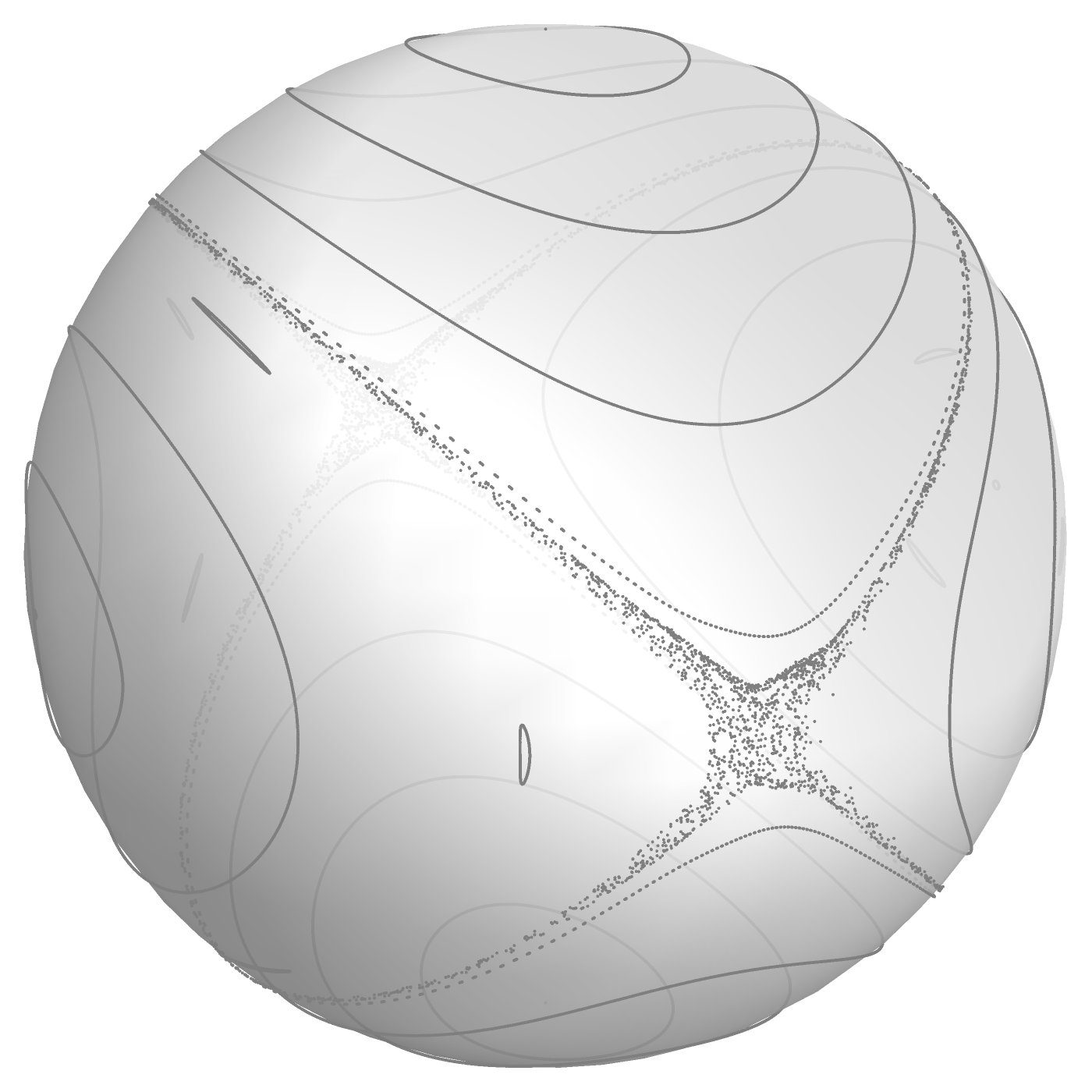}\qquad
	\includegraphics[width=.4\textwidth]{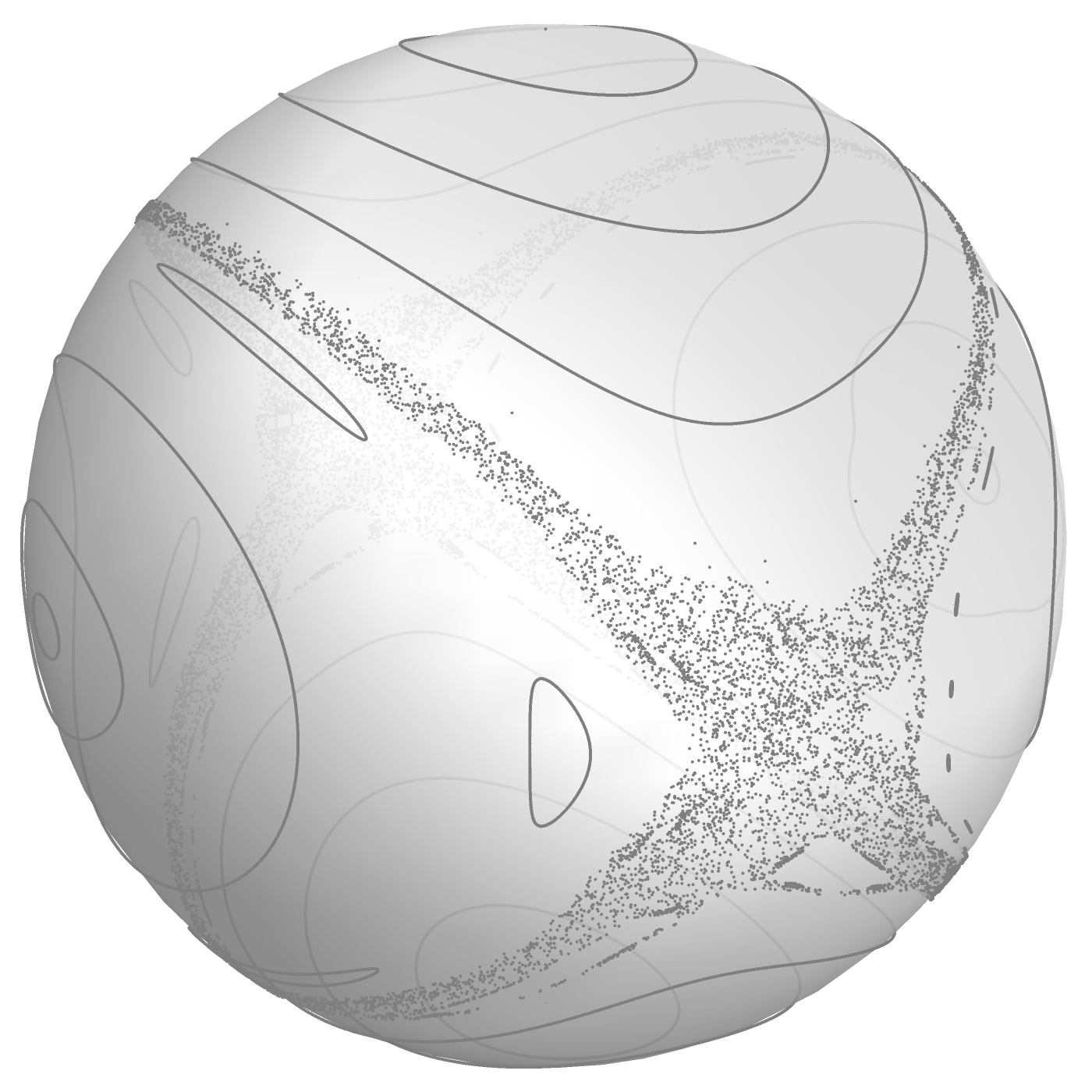}
	\caption{
		Poincaré section of the forced rigid body system with Hamiltonian~(\ref{eq:forced_rb_Ham}), approximated by the spherical midpoint method.
		Left: $\varepsilon=0.01$.
		Right: $\varepsilon=0.07$.
		Notice the development of chaos near the unstable equilibria points.
	}
	\label{forced_rb_fig}
\end{figure}


\subsection{4--particle system: point vortex dynamics on the sphere} 
\label{sub:4_particle_system_point_vortex_dynamics}

Point vortices constitute special solutions of the Euler fluid equations on two-dimensional manifolds; see the survey by~\citet{Ha2007} and references therein.
Consider the codimension zero submanifold of $(S^2)^n$ given by
\begin{equation}
	(S^2)^n_* \coloneqq \{ \w\in (S^2)^n; \w_i \neq \w_j,\; 1\leq i< j\leq n \}.
\end{equation}
Point vortex systems on the sphere, first studied by \citet{Bo1977}, are Hamiltonian systems on $(S^2)^{n}_{*}$ that provide approximate models for atmosphere dynamics with localised areas of high vorticity, such as cyclones on Earth and \emph{vortex streets}~\cite{HuMa2007} on Jupiter.
In absence of rotational forces, the Hamiltonian function is given by
\begin{equation}
	H(\w) = -\frac{1}{4\pi}\sum_{i<j} \kappa_i\kappa_j \ln(2-2\w_i\cdot\w_j).
\end{equation}
In this context, the constants $\kappa_i$ of the symplectic structure~\eqref{eq:symp_form} are called \emph{vortex strengths}.
The cases $n=1,2,3$ are integrable~\cite{KiNe1998,Sa1999}, but the case $n=4$ is non-integrable.
Characterisation and stability of relative equilibria have been studied extensively; see~\cite{LaMoRo2011} and references therein.

In this example, we study the case $n=4$ and $\kappa_i = 1$ by using the time-discrete approximation provided by the spherical midpoint method~\eqref{eq:area_midpoint}.
Our study reveals a non-trivial 4-dimensional invariant manifold of periodic solutions.\footnote{Interestingly, this special symmetric configuration was also found by \citet{LiMoRo2001}. We thank James Montaldi for pointing this out.}
The invariant manifold contains both stable and non-stable equilibria.

\begin{figure}
	\centering
	\input{fig_vortex_illustration}
	\caption{
		Illustration of the invariant submanifold $\bar{\mathcal{I}}\subset (S^2)^4_*$ given by~\eqref{eq:inv_submanifold_vortices}. 
	}
	\label{fig_inv_manifold}
\end{figure}
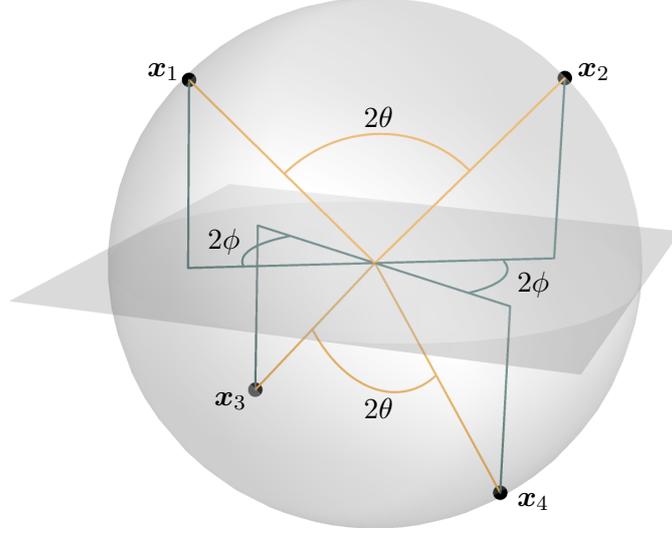

First, let $\vect c(\theta,\phi) \coloneqq \big(\cos(\phi)\sin(\theta),\sin(\theta)\sin(\phi),\cos(\theta)\big)$ and let
\begin{equation}
	\vect C(\theta,\phi) \coloneqq \begin{pmatrix}
		\vect c(\theta,\phi) \\
		\vect c(\theta,\phi+\pi) \\
		\vect c(\pi-\theta,-\phi) \\
		\vect c(\pi-\theta,\pi-\phi)
	\end{pmatrix} .
\end{equation}
Next, consider the two-dimensional submanifold of $(S^2)^4_*$ given by
\begin{equation} \label{eq:inv_submanifold_vortices}
	\bar{\mathcal{I}} = \{ \vect w\in (S^2)^4_*\,; \vect w = \vect C(\theta,\phi), 
		\theta \in [0,\pi), \phi\in [0,2\pi) \}.
\end{equation}
See \autoref{fig_inv_manifold} for an illustration.


\begin{figure}
	\centering
	\includegraphics[width=.45\textwidth]{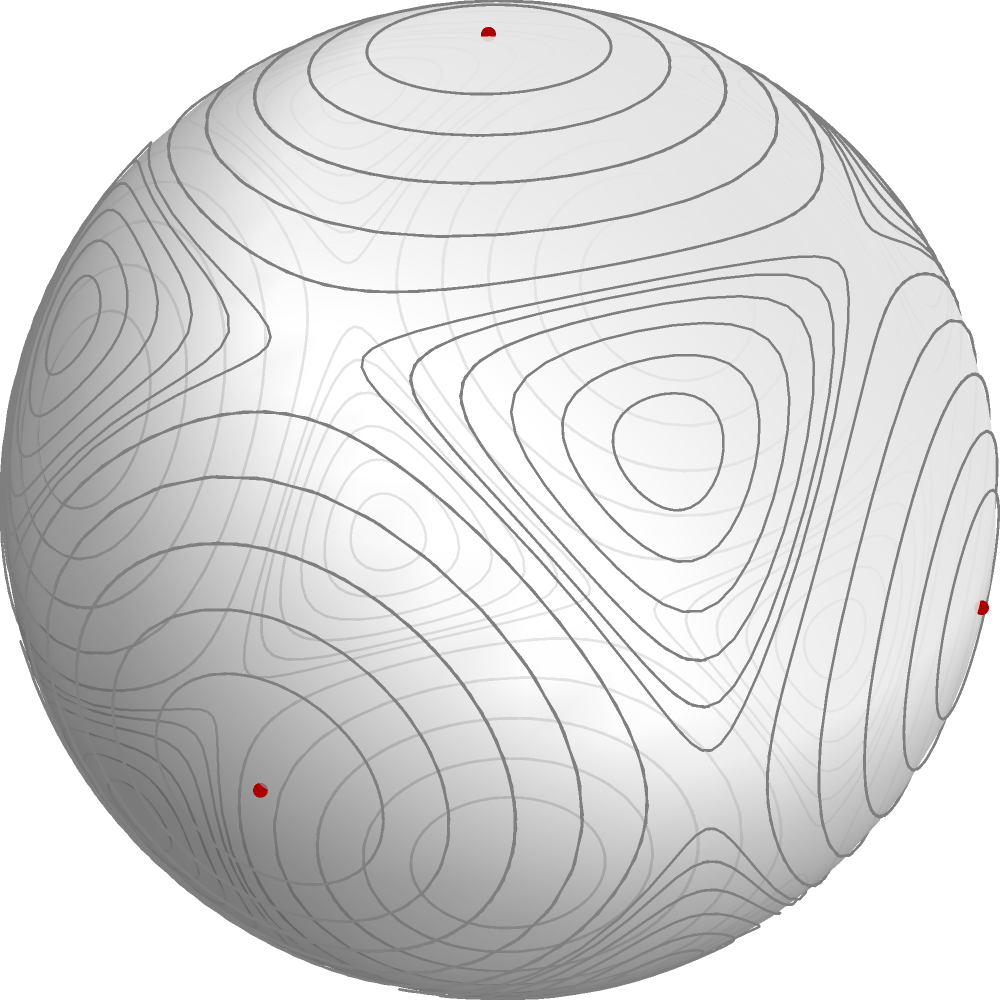}
	\caption{
		Particle trajectories on the invariant manifold $\bar{\mathcal{I}}$.
		The singular points are marked in red (these points are not part of $\bar{\mathcal{I}}$).
		Notice that there are two types of equilibria: the corners and the centres of the ``triangle like'' trajectories.
		The corners are unstable (bifurcation points) and the centres are stable (they are, in fact, stable on all of $(S^2)^4_*$, as is explained in~\cite{LaMoRo2011}).
	}
	\label{fig_vortex_phase_pic}
\end{figure}

The numerical observation that $\bar{\mathcal{I}}$ is an invariant manifold for the discrete spherical midpoint discretisation led us to the following result for the continuous system.

\begin{proposition}\label{prop:vortex_inv_manifold}
	\begin{equation}
		\mathcal{I} = \{ \vect w\in (S^2)^4_*; \vect w = A\cdot \bar{\vect w},\; A\in\SO(3), \bar{\vect w}\in \bar{\mathcal{I}} \}
	\end{equation}
	is a 5--dimensional invariant manifold for the continuous 4--particle point vortex system on the sphere with unitary vortex strengths.
	Furthermore, every trajectory on $\mathcal{I}$ is periodic.
\end{proposition}

\begin{proof}
	Direct calculations show that $X_H$ is tangent to $\bar{\mathcal{I}}$.
	The result for $\mathcal{I}$ follows since~$H$ is invariant with respect to the action of $\SO(3)$ on $(S^2)^4$.
\end{proof}

The example in this section illustrates how numerical experiments with a discrete symplectic model can give insight to the corresponding continuous system.
Generalisation of the result in \autoref{prop:vortex_inv_manifold} to other vortex ensembles is an interesting topic left for future studies.





\subsection{$n$--particle system: Heisenberg spin chain} 
\label{sub:n_particle_system_heisenberg_spin_chain}

The classical Heisenberg spin chain of micromagnetics is a Hamiltonian system on $(S^2)^n$ with strengths $\kappa_i=1$ and Hamiltonian
\begin{equation}\label{eq:nsys_Heisenberg_ham}
	H(\w) = \sum_{i=1}^{n}\w_{i-1}\cdot\w_{i} , \quad\w_0 = \w_n .
\end{equation}

For initial data distributed equidistantly on a closed curve, the system~\eqref{eq:nsys_Heisenberg_ham} is a space discrete approximation of the Landau--Lifshitz equation (see~\cite{La2011} for an overview).
This PDE is known to be integrable, so one can expect quasiperiodic behaviour in the solution.
Indeed, if we use the spherical midpoint method for~\eqref{eq:nsys_Heisenberg_ham}, with $n=100$ and initial data equidistantly distributed on a closed curve, the resulting dynamics appear to be quasiperiodic (see \autoref{fig_spin_chain}). 

The example in this section illustrates that the spherical midpoint method, together with a spatial discretisation, can be used to accurately capture the dynamics of integrable Hamiltonian PDEs on $S^2$.

\begin{figure}
	\centering
	\includegraphics[width=.4\textwidth]{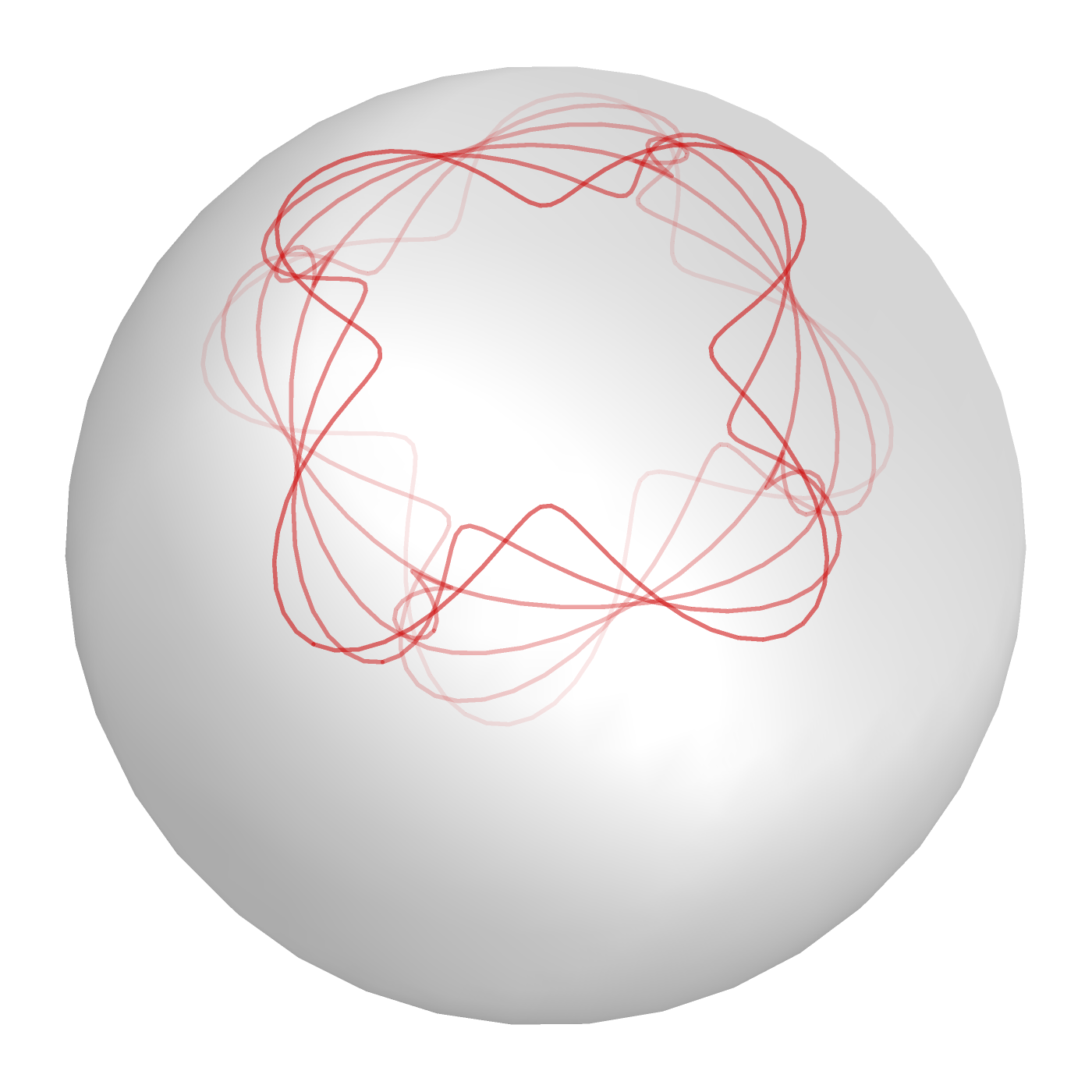}
	\caption{
		Evolution of the Heisenberg spin chain system~\eqref{eq:nsys_Heisenberg_ham} with $n=100$ for initial data equidistantly spaced on a simple closed curve using the spherical midpoint method.
		The corresponding Hamiltonian PDE (the Landau--Lifshitz equation) is known to be integrable.
	}
	\label{fig_spin_chain}
\end{figure}




\appendix
\section{Generalisation to Nambu systems}

It is natural to ask for which non-canonical symplectic or Poisson manifolds other than $(S^2)^n$ generating functions can be constructed. 
In full generality, this is an unsolved problem: no method is known to generate, for example, symplectic maps of a symplectic manifold $\vect F^{-1}(\{ 0\})$ in terms of  $\vect F\colon T^*\RR^d\to \RR^k$. 
In this appendix we shall show that the spherical midpoint method does generalise to Nambu mechanics~\cite{Na1973}. 
Let $C\colon\RR^3\to \RR$ be a homogeneous quadratic function defining the \emph{Nambu system} $\dot{\w} = \nabla C(\w) \times \nabla H(\w)$ with Hamiltonian $H\in C^\infty(\RR^3)$.
For $C(\w)=\frac{1}{2}\abs{\w}^2$, these are spin systems with a single spin. 
The Lagrange system $\dot w_1 = w_2 w_3$, $\dot w_2 = w_3 w_1$, $\dot w_3 = w_1 w_2$ is an example of a Nambu system with $C = \frac{1}{2}(w_1^2-w_2^2)$ and $H=\frac{1}{2}(w_1^2-w_3^2)$.

\begin{proposition} \label{thm:nambu}
	A symplectic integrator for the symplectic manifold given by the level set $C(\w) = c\ne 0$ in a Nambu system $\dot{\w} = \vect f = \nabla C \times \nabla H$, $C=\frac{1}{2}\w^T \vect C \w$, is given by the classical midpoint method applied to the Nambu system with Hamiltonian $H(\w/\sqrt{C(\w)/c})$.
\end{proposition}

\begin{proof}
	The Poisson structure of the Nambu system is given by $K(\w)(\aneutral,\bneutral) = \det([\vect C\w,\aneutral,\bneutral])$. 
	Let $\f$ be the projected Nambu vector field. 
	Calculations as in the proof of \autoref{thm:main} now give
	\begin{align} 
		K(\wplus)(\aplus,\bplus) - K(\wminus)(\aminus,\bminus) &= h^3 \det([\vect C\wdif,\adif,\bdif] \\
		& = h^3 \det([\vect C \f(\wmid), -D\f(\wmid)^{\top}\amid,-D\f(\wmid)^{\top}\bmid]).
	\end{align}
	As before, all three arguments are orthogonal to $\wmid$: $\vect C\f(\wmid)$, because $\f(\wmid)$ is tangent to the level set $C(\w)=c$, whose normal at $\wmid$ is $\vect C\wmid$, and $-D\f(\wmid)^\top\amid$ because $\langle -D\f(\wmid)^\top\amid,\wmid\rangle = \langle \amid,-D\f(\wmid)\wmid\rangle$, and because $\w \mapsto H(\w/\sqrt{C(\w)/c})$ is homogeneous on rays, $\f$ is constant on rays.
\end{proof}

Note that if $H$ is also a homogeneous quadratic (as in the Lagrange system), then the method preserves $C$ and $H$ and generates an integrable map.
The Nambu systems in \autoref{thm:nambu} are all 3-dimensional Lie--Poisson systems. 
There are 9 inequivalent families of real irreducible 3-dimensional Lie algebras \cite{Pa1976}. 
Five of them have homogeneous quadratic Casimirs and are covered by \autoref{thm:nambu}: in the notation of \cite{Pa1976}, they are
$A_{3,1}$ ($C=w_1^2$, Heisenberg Lie algebra)
$A_{3,4}$ ($C=w_1 w_2$, $\mathfrak{e}(1,1)$);
$A_{3,6}$ ($C=w_1^2+w_2^2$, $\mathfrak{e}(2)$);
$A_{3,8}$ ($C=w_2^2 + w_1 w_3$, $\mathfrak{su}(1,1)$, $\mathfrak{sl}(2)$);
$A_{3,9}$ ($C=w_1^2+w_2^2+w_3^2$, $\mathfrak{su}(2)$, $\mathfrak{so}(3)$).
A large set of Lie--Poisson systems is obtained by direct products of the duals of these Lie algebras.
Such a structure was already mentioned by Nambu in his original paper, noting the application to spin systems.
The spherical midpoint method applies to these systems; it generates symplectic maps in neighbourhoods of symplectic leaves with $c \neq 0$. 


\bibliographystyle{../amsplainnat} 
\bibliography{../collective}

\end{document}

%% file: definitions.tex
\usepackage{amsmath,amssymb}

\providecommand{\abs}[1]{\lvert#1\rvert}
\providecommand{\vect}[1]{\boldsymbol{#1}}

\newcommand{\ud}{\mathrm{d}}

\newcommand{\pd}{\partial}

\newcommand{\R}{{\mathbb R}}
\newcommand{\CC}{{\mathbb C}}

\newcommand{\Xcal}{\mathfrak{X}}

\newcommand*\SO{\mathrm{SO}}
\newcommand*\so{\mathfrak{so}}
\newcommand*\SU{\mathrm{SU}}
\newcommand*\su{\mathfrak{su}}

\newcommand{\Rattle}{{\small RATTLE}}

%% file: fig_symplectic_leaves.tex
\begin{tikzpicture}
	\node[anchor=south west, inner sep=0] (image) at (0,0) {\includegraphics[width=0.36\textwidth]{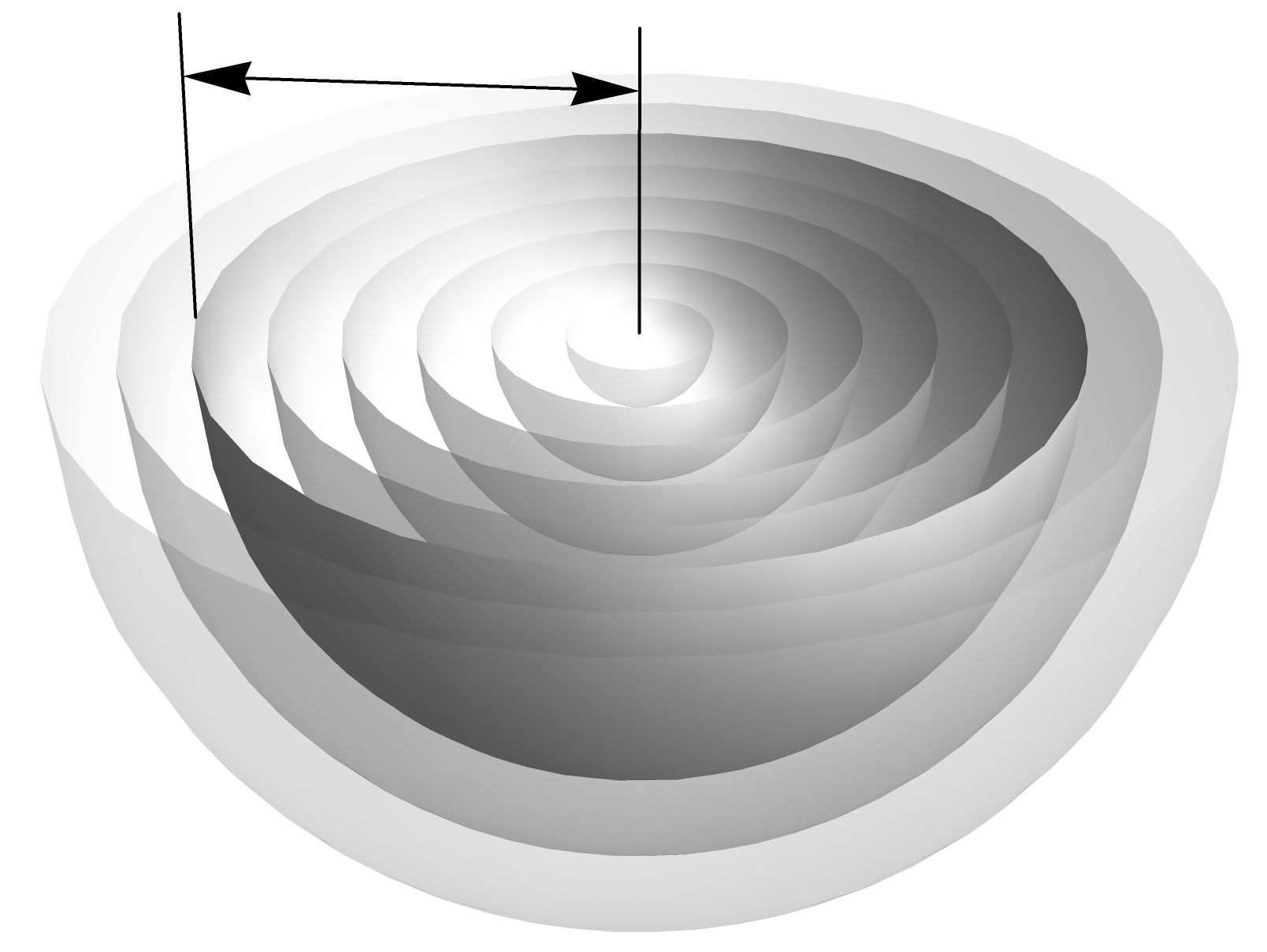}};
	\begin{scope}[x={(image.south east)},y={(image.north west)}]
		\coordinate (label1) at (0.32,0.92) {};
		\node[above=0ex,rotate=0] at (label1) {$\lambda$};
	\end{scope}
\end{tikzpicture}

%% file: fig_symrb_error_plot.tex
\begin{tikzpicture}
	\node[anchor=south west, inner sep=0] (image) at (0,0) {\includegraphics[width=0.8\textwidth]{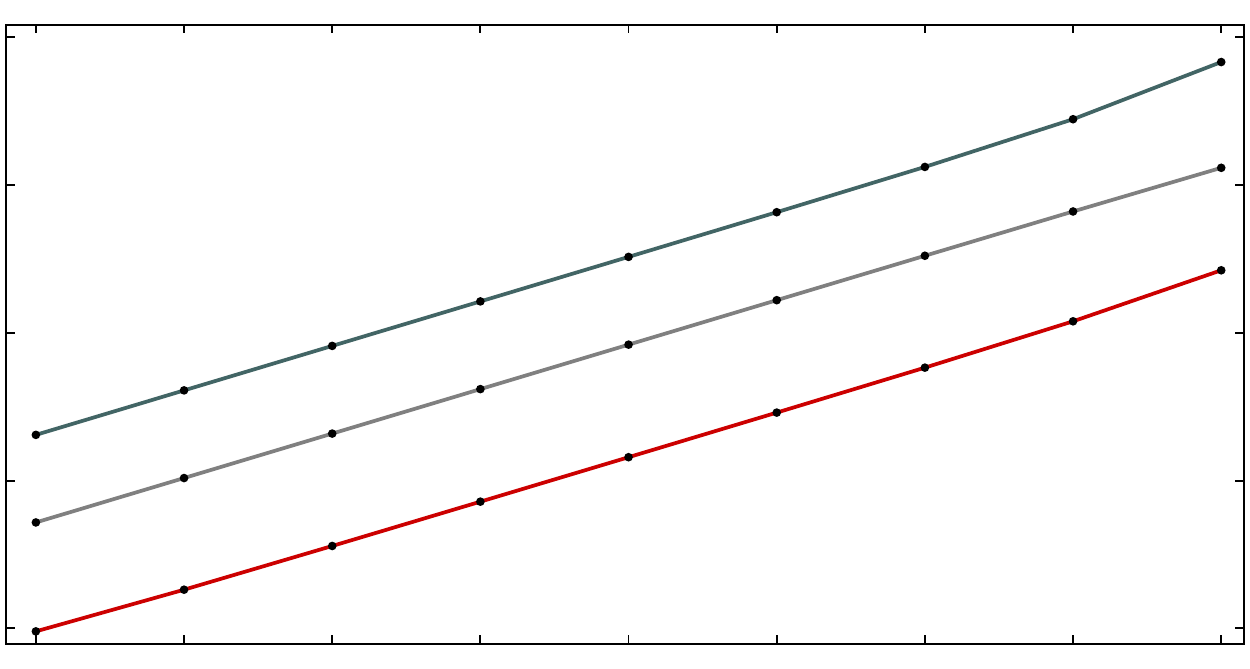}};
	\begin{scope}[x={(image.south east)},y={(image.north west)}]
		\coordinate (x1) at (0.04,0.02) {};
		\coordinate (x2) at (0.275,0.02) {};
		\coordinate (x3) at (0.513,0.02) {};
		\coordinate (x4) at (0.75,0.02) {};
		\coordinate (x5) at (0.985,0.02) {};
		\node[below=-1ex] at (x1) {\small$2^{-9}$};
		\node[below=-1ex] at (x2) {\small$2^{-7}$};
		\node[below=-1ex] at (x3) {\small$2^{-5}$};
		\node[below=-1ex] at (x4) {\small$2^{-3}$};
		\node[below=-1ex] at (x5) {\small$2^{-1}$};
		\coordinate (y5) at (0,0.95) {};
		\coordinate (y4) at (0,0.73) {};
		\coordinate (y3) at (0,0.51) {};
		\coordinate (y2) at (0,0.29) {};
		\coordinate (y1) at (0,0.069) {};
		\node[left=-0.7ex] at (y5) {\small$10^{0\phantom{-}}$};
		\node[left=-0.7ex] at (y4) {\small$10^{-2}$};
		\node[left=-0.7ex] at (y3) {\small$10^{-4}$};
		\node[left=-0.7ex] at (y2) {\small$10^{-6}$};
		\node[left=-0.7ex] at (y1) {\small$10^{-8}$};
		\coordinate (exp) at (0,1.02) {};
		\coordinate (xlabel) at (0.5,0) {};
		\coordinate (ylabel) at (-0.08,0.5) {};
		\node[below=1ex] at (xlabel) {\footnotesize Time-step length $h$};	
		\node[above,rotate=90] at (ylabel) {\footnotesize Maximum error};	
		\coordinate (label1) at (0.67,0.71) {};
		\coordinate (label2) at (0.67,0.57) {};
		\coordinate (label3) at (0.67,0.4) {};
		\node[above=0ex,rotate=17] at (label1) {Moser--Veselov};
		\node[above=0ex,rotate=17] at (label2) {Classical midpoint};
		\node[above=0ex,rotate=17] at (label3) {Spherical midpoint};
	\end{scope}
\end{tikzpicture}

%% file: fig_nonsymrb_energy_drift.tex
\begin{tikzpicture}
	\node[anchor=south west, inner sep=0] (image) at (0,0) {\includegraphics[width=0.96\textwidth]{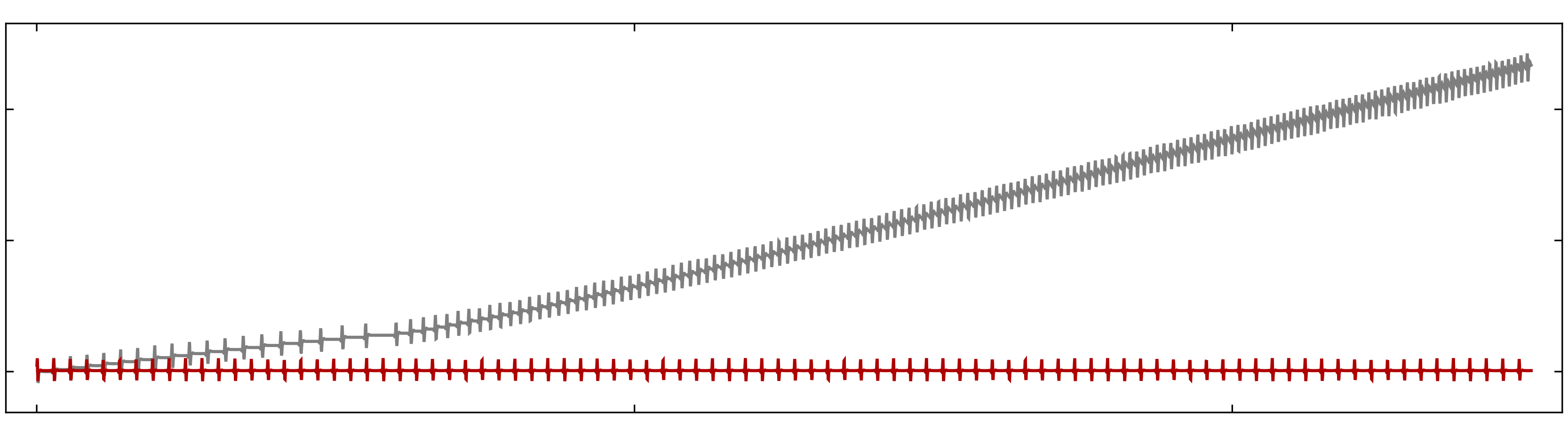}};
	\begin{scope}[x={(image.south east)},y={(image.north west)}]
		\coordinate (x1) at (0.025,0.02) {};
		\coordinate (x2) at (0.403,0.02) {};
		\coordinate (x3) at (0.785,0.02) {};
		\node[below=-1ex] at (x1) {\small$0$};
		\node[below=-1ex] at (x2) {\small$1000$};
		\node[below=-1ex] at (x3) {\small$2000$};
		\coordinate (y1) at (0,0.75) {};
		\coordinate (y2) at (0,0.452) {};
		\coordinate (y3) at (0,0.152) {};
		\node[left=-0.5ex] at (y1) {\small$4$};
		\node[left=-0.5ex] at (y2) {\small$2$};
		\node[left=-0.5ex] at (y3) {\small$0$};
		\coordinate (exp) at (0,1.02) {};
		\node[right=-0.5ex] at (exp) {\small$\times 10^{-2}$};	
		\coordinate (label1) at (0.67,0.6) {};
		\coordinate (label2) at (0.67,0.16) {};
		\node[above=0ex,rotate=14] at (label1) {Classical midpoint};
		\node[above=0ex] at (label2) {Spherical midpoint};
	\end{scope}
\end{tikzpicture}

%% file: fig_vortex_illustration.tex
\begin{tikzpicture}
	\node[anchor=south west, inner sep=0] (image) at (0,0) {\includegraphics[width=0.6\textwidth]{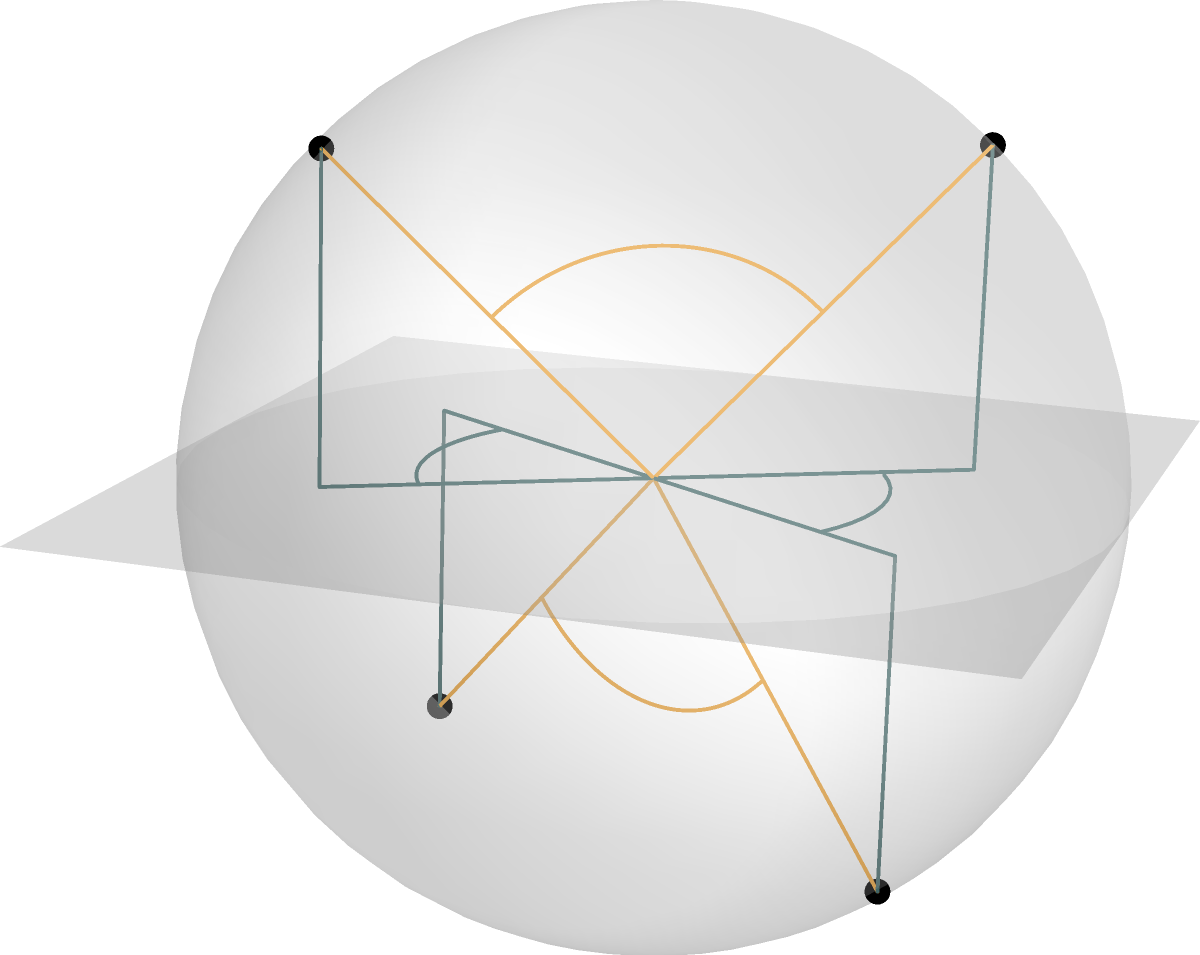}};
	\begin{scope}[x={(image.south east)},y={(image.north west)}]
		\coordinate (theta1) at (0.55,0.76) {};
		\coordinate (theta2) at (0.55,0.24) {};
		\coordinate (phi1) at (0.34,0.54) {};
		\coordinate (phi2) at (0.76,0.46) {};
		\node[above=-1ex] at (theta1) {$2\theta$};
		\node[below=-1ex] at (theta2) {$2\theta$};
		\node[left=-1ex] at (phi1) {$2\phi$};
		\node[right=-1ex] at (phi2) {$2\phi$};
		\coordinate (x1) at (0.25,0.86) {};
		\coordinate (x2) at (0.85,0.86) {};
		\coordinate (x3) at (0.35,0.24) {};
		\coordinate (x4) at (0.76,0.05) {};
		\node[left=-1ex] at (x1) {$\vect{x}_1$};
		\node[right=-1ex] at (x2) {$\vect{x}_2$};
		\node[left=-1ex] at (x3) {$\vect{x}_3$};
		\node[right=-1ex] at (x4) {$\vect{x}_4$};
		\coordinate (x2) at (0.403,0.02) {};
		\coordinate (x3) at (0.785,0.02) {};
		\coordinate (y1) at (0,0.75) {};
		\coordinate (y2) at (0,0.452) {};
		\coordinate (y3) at (0,0.152) {};
		\coordinate (exp) at (0,1.02) {};
		\coordinate (label1) at (0.67,0.6) {};
		\coordinate (label2) at (0.67,0.16) {};
	\end{scope}
\end{tikzpicture}